\documentclass[11pt]{article}
\usepackage{amsmath}
\usepackage{amsthm}
\usepackage{amssymb}
\usepackage{url}
\usepackage[final]{showkeys}
\usepackage[usenames,dvipsnames]{xcolor}
\usepackage{graphicx}
\usepackage{bm}

\usepackage[margin=1in]{geometry}

\usepackage{titlesec}

\titlespacing*{\section}{0pt}{11pt plus 2pt minus 2pt}{8pt plus 1pt minus 1pt}
\titlespacing*{\subsection}{0pt}{9pt plus 2pt minus 2pt}{6pt plus 1pt minus 1pt}

\usepackage[vlined]{algorithm2e}

\usepackage{fullpage}

\usepackage{stmaryrd}

\usepackage[colorlinks=true, allcolors=blue]{hyperref}

\usepackage{latexsym}

\newtheorem{theorem}{Theorem}[section]
\newtheorem{lemma}[theorem]{Lemma}

\newtheorem{corollary}[theorem]{Corollary}
\newtheorem{definition}[theorem]{Definition}

\newcommand{\AppendixBanner}{%
  \hbox to \textwidth{%
    \leaders\hbox{\rule[.55ex]{1em}{0.5pt}}\hfill
    \enspace\textbf{APPENDIX}\enspace
    \leaders\hbox{\rule[.55ex]{1em}{0.5pt}}\hfill
  }%
}

\makeatletter
\def\ps@appendix{%
  \let\@mkboth\markboth
  \def\@oddhead{\AppendixBanner}%
  \def\@evenhead{\AppendixBanner}%
  \def\@oddfoot{\hfil\thepage\hfil}%
  \def\@evenfoot{\hfil\thepage\hfil}%
}
\makeatother

\usepackage[T1]{fontenc}    
\usepackage[utf8]{inputenc} 
\usepackage{textcomp}       

\DeclareUnicodeCharacter{0301}{\'} 

\usepackage[
    style=alphabetic, 
    backend=biber, 
    giveninits=false, 
    maxbibnames=99, 
    maxcitenames=3, 
    maxalphanames=4,  
    minalphanames=3,
    doi=true, 
    url=true
]{biblatex}
\DeclareSourcemap{
  \maps[datatype=bibtex]{
    \map{
      \step[fieldsource=doi, final]
      \step[fieldset=url, null]
    }
  }
}

\renewbibmacro{in:}{%
  \ifentrytype{article}{}{\printtext{\bibstring{in}\intitlepunct}}%
}

\DeclareFieldFormat[article,incollection,inproceedings]{title}{\textit{#1}}

\newcommand{\sm}{\setminus}

\newcommand{\EQ}{\;=\;}
\newcommand{\LE}{\;\le\;}
\newcommand{\LT}{\;<\;}
\newcommand{\GE}{\;\ge\;}
\newcommand{\GT}{\;>\;}

\newcommand{\ba}{\bm{a}}
\newcommand{\bb}{\bm{b}}
\newcommand{\bc}{\bm{c}}
\newcommand{\bs}{\bm{s}}

\newcommand{\bu}{\bm{u}}
\newcommand{\bv}{\bm{v}}

\newcommand{\RR}{\mathbb{R}}

\newcommand{\vH}{\vec{H}}
\newcommand{\sink}{sink}
\newcommand{\pra}{\prod_{i=1}^n A_i}
\newcommand{\praa}{\prod_{i=1}^{n_0} A_i}

\newcommand{\algcall}[1]{\text{\normalfont\textsc{#1}}}
\newcommand{\RAR}{\emph{Random-Action-Removal}}
\newcommand{\RARalg}{\algcall{Rand-Act-Rem}}
\newcommand{\cH}{\vec{\mathcal{H}}_{\text{AUSO}}(\pra)}
\newcommand{\oH}{\vec{\mathcal{H}}(\pra)}

\newcommand{\s}{\ell}

\AtBeginDocument{
  \setlength{\abovedisplayskip}{7pt plus 2pt minus 2pt}
  \setlength{\belowdisplayskip}{7pt plus 2pt minus 2pt}
  \setlength{\abovedisplayshortskip}{2pt plus 2pt minus 2pt}
  \setlength{\belowdisplayshortskip}{4pt plus 2pt minus 2pt}
}

\title{\vspace*{-0.5cm}Improved subexponential analysis of the \RAR\ algorithm for 2-player turn-based games and non-binary AUSOs}

\author{Uri Zwick\thanks{Blavatnik School of
Computer Science, Tel Aviv University,
  Israel. E-mail: {\tt
    zwick@tau.ac.il}. Work supported by ISF grant no.\ 2854/20.}}
\date{}

\begin{document}

\maketitle

\vspace*{-1cm}
\begin{abstract}%
\setlength{\parindent}{0pt}%
\setlength{\parskip}{3pt plus 2pt}%
    \noindent
    We give a concise description and an improved analysis of the \RAR\ algorithm for solving 2-player, 0-sum, turn-based, possibly infinite duration, stochastic or non-stochastic games played on graphs, or on finite sets of states. More generally, the algorithm can be used to find the \emph{sink} of an Acyclic Unique Sink Orientation (AUSO) of a non-binary hypercube. The families of games that can be solved by the algorithm include discounted and non-discounted stochastic games (SGs) and Mean Payoff Games (MPGs). The obtained algorithm is the fastest known randomized algorithm for solving such games, slightly improving on a much more complicated algorithm of Hansen and Zwick (STOC 2015).

    The \RAR\ algorithm is an adaptation of the \emph{Random-Facet} algorithm used to solve linear programming (LP) problems, or, more generally, LP-type problems. Two dual variants of the Random-Facet algorithm were developed independently by Kalai (STOC 1992) and by Matou{\v{s}}ek, Sharir and Welzl (SoCG 1992). For LP problems, the algorithm of Kalai is a primal \emph{simplex} algorithm, while the algorithm of Matou{\v{s}}ek, Sharir and Welzl is a dual \emph{simplex} algorithm. The \RAR\ algorithm for games or AUSOs is an adaptation of the dual algorithm of Matou{\v{s}}ek, Sharir and Welzl, and is a randomized \emph{strategy iteration} algorithm.

    We present an improved analysis of the \RAR\ algorithm that shows that it solves games with~$n$ states and $m\ge 2n$ actions, belonging to one of the players, in $e^{O(\sqrt{n\ln(m/n)}\,)}$ 
    time. This improves on a previous $e^{O(\sqrt{n\ln(m/\sqrt n)}\,)}$ bound that follows from the analysis of Matou{\v{s}}ek, Sharir and Welzl (SoCG 1992). Our improved analysis relies of the \emph{hypercube} structure of games and AUSOs. It does not improve on the bound of Matou{\v{s}}ek, Sharir and Welzl for general LP-type problems that may lack such structure. The improved bound is obtained by deriving a tighter recurrence relation implied by the hypercube structure, and by using interesting mathematical techniques, such as the discrete WKB method, for solving the resulting non-standard recurrence relation. 

    An $e^{O(\sqrt{n\ln(m/n)}\,)}$ bound for games and AUSOs also follows from an improved algorithm for solving LP and LP-type problems of Hansen and Zwick (STOC 2015). However, this bound is obtained using a much more complicated algorithm, and has worst constant factors in the exponent. The main point of this paper is that the $e^{O(\sqrt{n\ln(m/n)}\,)}$ bound holds for the simple and natural \RAR\ algorithm without any modification.

    The improvement from $e^{O(\sqrt{n\ln(m/\sqrt n)}\,)}$ to $e^{O(\sqrt{n\ln(m/n)}\,)}$ is especially significant when $m=O(n)$, which is a very natural setting, as then the bound on the running time of the algorithm drops from $e^{O(\sqrt{n\ln n})}$ to $e^{O(\sqrt{n})}$.
    
\end{abstract}

\section{Introduction}

The simplex algorithm is one of the most important algorithms used to solve linear programming (LP) problems, both in theory and in practice. However, no variant of the simplex algorithm is known to run in polynomial time in the worst case. Furthermore, most natural variants of the simplex algorithm are known to require exponential time on some carefully designed instances. 

Kalai \cite{Kalai92,Kalai92b} and Matou{\v{s}}ek, Sharir and Welzl \cite{MaShWe96}, in breakthrough results, obtained randomized variants of the simplex algorithm with \emph{subexponential} running times. Their algorithms are actually duals of each other. The algorithm of Kalai \cite{Kalai92,Kalai92b}, usually called the \emph{Random-Facet} algorithm, is a \emph{primal} simplex algorithm, while the algorithm of Matou{\v{s}}ek, Sharir and Welzl \cite{MaShWe96} is a \emph{dual} simplex algorithm. 

Matou{\v{s}}ek, Sharir and Welzl \cite{MaShWe96} point out that their algorithm works not only for LPs but also for a wider class of LP-type problems. The first connection between their algorithm and the games considered in this paper was made by Ludwig \cite{Ludwig95} who obtained an $e^{O(\sqrt{n})}$-time algorithm for solving \emph{binary} games, i.e., games in which one of the players controls $n$ states and has exactly two actions from each state. {Bj{\"o}rklund} and Vorobyov \cite{BjVo05,BjVo07} extended the algorithm to general non-binary games and used the analysis of Matou{\v{s}}ek, Sharir and Welzl \cite{MaShWe96} to obtain a bound of $e^{O(\sqrt{n\ln(m/\sqrt n)}\,)}$ on its running time, where~$m$ is the total number of actions.

The algorithm of {Bj{\"o}rklund} and Vorobyov \cite{BjVo05,BjVo07}, which for reasons that will become clear shortly we call the \RAR\ algorithm, is simply a recasting of the algorithm of Matou{\v{s}}ek, Sharir and Welzl \cite{MaShWe96} in the context of games or AUSOs. Halman \cite{Halman07} showed that all games considered are LP-type problems, which explains why such a recasting is possible. {Bj{\"o}rklund} and Vorobyov \cite{BjVo05,BjVo07} rely completely on the probabilistic analysis of Matou{\v{s}}ek, Sharir and Welzl \cite{MaShWe96} that holds for general LP-type problems.

Games and AUSOs have more structure than general LP-type problems. Generally speaking, the number of \emph{states} of a game corresponds to the \emph{dimension} of a (basis-regular) LP-type problem. An \emph{action} of a game corresponds to a \emph{constraint}, or an \emph{inequality}, of the LP-type problem. A \emph{strategy} of a game corresponds to a \emph{basis} of the LP-type problem. In games and AUSOs, each action belongs to a specific state. Furthermore, each strategy uses exactly one action of each state. We refer to this as the \emph{hypercube} structure of games and AUSOs. General LP-type problems usually lack such structure. In particular, there is no fixed partition of the set of constraints such that each basis contains exactly one constraint from each part of the partition. We use the hypercube structure of games to obtain an improved bound on the running time of the \RAR\ algorithm. As explained, for games with $m=O(n)$ actions, the bound is improved from $2^{O(\sqrt{n\ln n})}$ to $2^{O(\sqrt{n})}$, providing a smooth transition of the complexity of the algorithm when moving from binary to non-binary games.

In a nutshell, the \RAR\ algorithm works as follows. It starts with an arbitrary strategy of one of the players. It chooses a \emph{random action} not contained in the current strategy, \emph{removes} it from the game and solves the resulting subgame recursively. If the removed action cannot be used to improve the optimal strategy of the subgame, then this strategy is also the optimal strategy of the whole game. Otherwise, the algorithm performs an improvement step using the removed strategy and performs a recursive call, on the whole game, starting with the improved strategy. The properties of the games studied imply that the \RAR\ algorithm must always terminate with an optimal strategy. The main question is how many improvement steps it performs, in expectation, before it reachs an optimal strategy. Each improvement step takes only polynomial time. As mentioned, the analysis of Matou{\v{s}}ek, Sharir and Welzl \cite{MaShWe96} shows that the expected number of iterations performed by the algorithm on any game with $n$ states and $m$ actions in total, for one of the players, is at most $e^{O(\sqrt{n\ln(m/\sqrt n)}\,)}$.

We perform an improved analysis of the \RAR\ algorithm that utilizes the hypercube structure of the games considered. (Our improved analysis does not hold for general LP-type problems that do not have a hypercube structure.) This leads to an improved bivariate recurrence relation that bounds the expected number of steps performed by the algorithm in terms of~$n$, number of states and~$m$, the total number of actions for one of the players in the input game. 
The improved recurrence relation utilizes cases in which the number of effective states of the game, and not just the total number of actions, decrease as a result of some improvement steps performed by the algorithm.

The main technical challenge is then solving the new recurrence relation. Matou{\v{s}}ek, Sharir and Welzl \cite{MaShWe96} solve their recurrence relation using generating functions techniques. There does not seem to be an easy way of using their generating function techniques to solve our modified recurrence relation. We thus use a different approach. 

We perform an asymptotic analysis of our recurrence relation using a technique similar to the discrete Wentzel–Kramers–Brillouin (WKB) method. (See, e.g., Bender and Orszag \cite{BenderOrszag1999}.) In this method we stipulate an \emph{ansatz}, an assumed general form of the solution, sometimes obtained using numerical experiments. By manipulating the ansatz we obtain a differential equation, known as the \emph{eikonal} equation, that dictates the asymptotic behavior of the solution. We then use properties of the eikonal differential equation to show that a tuned version of the ansatz is indeed an upper bound on the exact solution of the recurrence relation. 

The difficulty of applying the WKB technique in our case is that the eikonal equation does not seem to have an explicit closed-form solution. We thus need to work with the implicit solution defined by the differential equation. Nevertheless, we can obtain rigorous lower and upper bounds on the exact solution of the eikonal equation that lead to rigorous upper bounds on the solution of our bivariate recurrence relation. Some of the steps used in this process were devised with the help of Google's Gemini and OpenAI's ChatGPT. However, the resulting analysis is relatively short, completely explicit, and can be easily verified by traditional methods, without the use of a computer.

Since our improved analysis of the \RAR\ algorithm only relies on the hypercube structure of the games considered, the analysis actually applies to a wider class of abstract games, best described by the notion of \emph{Acyclic Unique Sink Orientations (AUSOs)} of (non-binary) hypercubes. Each concrete game considered gives rise to such an AUSO and finding an optimal strategy for one of the players in the game reduces to finding the \emph{sink} of this AUSO. Szab{\'o} and Welzl \cite{SzWe01} began a systematic study of binary AUSOs (and USOs, dropping the acyclicity requirement), though the concept goes back much further, e.g., to Williamson Hoke \cite{WiHo88}. G{\"a}rtner, Morris, and R{\"u}st \cite{GaMoRu08} studies non-binary USOs. G{\"a}rtner \cite{Gartner02} studies the behavior of the \RAR\ algorithm on binary AUSOs. Our analysis is a far-reaching extension of his analysis to non-binary AUSOs. The non-binary case is much more challenging than the binary case and requires the development of new techniques.

\subsection*{Paper organization}

The rest of this paper is organized as follows. In Appendix~\ref{A:context} we give additional high-level details about the games considered in this paper, about known algorithms for solving them and their relation to algorithms for solving linear programming problems, and about AUSOs. We also explain the significance of the questions considered and of the obtained results.

In Section~\ref{sec:AUSOs} we briefly define the notion of AUSOs. A more formal definition appears in Appendix~\ref{A:AUSOs}. In Section~\ref{sec:Games} we define the family of 2-player Turn-Based Stochastic Games (2-TBSG) which is the most general family of concrete games considered in this paper. It includes, in particular, the family of Simple Stochastic Games (SSGs). Mean Payoff Games (MPGs) and Discounted Payoff Games (DPGs) are deterministic versions of 2-TBSGs. We also define the family of Strategy Improvement (SI) algorithms that can be used to solve these games. The \RAR\ is a simple randomized member of this family. We also explain how the solution of a 2-TBSG can be reduced to the problem of finding the sink of an appropriate AUSO. In Section~\ref{sec:Random-Action} we formally define the \RAR\ algorithm as applied to AUSOs. 

The main contributions of this paper are contained in Sections~\ref{sec:Analysis} to~\ref{S:Recurrence}. In Section~\ref{sec:Analysis} we obtain an improved analysis of the \RAR\ algorithm when applied to non-binary AUSOs. The main result of this section is an improved bivariate recurrence relation bounding the expected number of improvement steps performed by the algorithm. The two parameters of this recurrence relation are~$n$, the number of states, and~$m$, the total number of actions. As any non-degenerate state has at least two actions associated with it, we can assume that $m\ge 2n$. It is convenient to replace the parameter~$m$ by $k=m-2n\ge 0$. In a binary game, each state has exactly two actions and $m=2n$, so $k=0$. The parameter~$k$ thus represents the \emph{excess} number of actions with respect to a binary game. In Appendix~\ref{A:equiv} we obtain an extremely simple equivalent recurrence relation. 

In Section~\ref{sec:Binary} we analyze the $k=0$ case, corresponding to binary games, very slightly extending the results of G{\"a}rtner \cite{Gartner02} who showed that the expected number of steps performed by the algorithm on binary games is at most $e^{2\sqrt{n}}$. In Section~\ref{S:Simple} we extend the $e^{2\sqrt{n}}$ bound for binary games to an $e^{2\sqrt{n+k}}=e^{2\sqrt{m-n}}$ bound for non-binary games. This result already implies an upper bound of $e^{O(\sqrt{n})}$ for games with $m=O(n)$, i.e., with a linear total number of actions. Interestingly, the bound $e^{2\sqrt{n+k}}$ seems to be asymptotically tight for $k\le n$, or equivalently for $m\le 3n$. As $k$ grows beyond~$n$, the $e^{2\sqrt{n+k}}$ upper bound becomes more and more loose. 

In Section~\ref{S:Recurrence}, the main technical section of the paper, we obtain an upper bound of $e^{2\sqrt{n\Phi(\frac{k}{n})}}$ on the expected number of improvement steps perfored by the \RAR\ algorithm, where $\Phi(\alpha)$ satisfies an explicitly derived differential equation. We also show that $\Phi(\alpha)\le 2+\ln\frac{1+\alpha}{2}+\ln(1+
\ln\frac{1+\alpha}{2})$, for every $\alpha\ge 1$. This justifies the $e^{O(\sqrt{n\ln\frac{m}{n}})}$ upper bound on the running time of the algorithm claimed in the abstract and introduction.

\section{Acyclic Unique Sink Orientations (AUSOs)}\label{sec:AUSOs}

We give here a brief definition of Acyclic Unique Sink Orientations (AUSOs) of hypercubes. More formal definitions are given in Appendix~\ref{A:AUSOs}. If $A_1,A_2,\ldots,A_n$ be disjoint finite sets, the \emph{hypercube} $H(\pra)$ is a graph whose vertex set is $\pra$ with edges connecting two vertices that differ in exactly one coordinate. If $|A_i|=2$, for $i\in[n]$, then the hypercube is said to be \emph{binary}. An \emph{orientation} $\vH$ of $H(\pra)$ is obtained by giving each edge of the hypercube a direction. A \emph{sink} of an orientation $\vH$ is a vertex $\bs\in\pra$ such that all edges adjacent to~$\bs$ are directed towards~$\bs$. An orientation $\vH$ is acyclic if it does not contain a directed cycle. If $A'_i\subseteq A_i$, for $i\in[n]$, then the \emph{subcube} $\vec{H}[\prod_{i=1}^n A'_i]$ is the oriented hypercube on $\prod_{i=1}^n A'_i$ induced by~$\vec{H}$. Finally, an orientation $\vH$ is said to be an \emph{Acyclic Unique Sink Orientation (AUSO)} if and only if it is acyclic and any subcube of~$\vH$, including $\vH$ itself, has a unique sink. We let $\vec{\mathcal{H}}_{AUSO}(\pra)$ the set of all the AUSOs defined on $\pra$.

AUSOs provide an elegant abstraction of optimization problems in which local optima are guaranteed to be global optima, and for which this also holds for any subproblem. An edge $\ba\to\bb$ indicates that $\bb$ is better than $\ba$. Noteworthy examples are linear programs whose feasibility region is a combinatorial hypercube, e.g., the Klee-Minty cubes \cite{KaKl92} used to obtain lower bounds for various pivoting rules for the simplex algorithm, or the games that we consider in the next section. For more on AUSOs, see \cite{SzWe01,Matousek94}.

A sink-finding algorithm is given \emph{oracle} access to an input AUSO $\vH$. The algorithm can query a vertex $\ba\in\vH$ and receive the direction of all the edges incident on~$\ba$. The complexity of the algorithm is the (expected) number of queries performed by the algorithm before the sink of $\vH$ is identified. It is a major open problem whether there is a sink-finding algorithm that performs a polynomial, in~$n$ and $m=\sum_{i=1}^n|A_i|$, number of queries. The fastest known sink-finding algorithm is the \RAR\ algorithm analyzed in this paper. As mentioned, we improve the bound on the complexity of this algorithm from $e^{O\bigl(\sqrt{n\ln\frac{m}{\sqrt n}}\,\bigr)}$ to the more natural bound of $e^{O\bigl(\sqrt{n\ln\frac{m}{n}}\,\bigr)}$.

\section{2-player, 0-sum, Turn-Based, Stochastic Games (2-TBSGs)}\label{sec:Games}

\newcommand{\obj}{{\text{\scshape obj}}}
\newcommand{\agg}{{\text{\scshape agg}}}
\newcommand{\avg}{{\text{\scshape avg}}}
\newcommand{\dsc}{{\text{\scshape dsc}}}
\newcommand{\VAL}{{\text{VAL}}}

We give here a brief description of the families of games considered in this paper. For more formal definitions and further details, see Appendix~\ref{A:Games}. A \emph{2-player 0-sum Turn-Based Stochastic Game (2-TBSG)} is a perfect-information, infinite-duration game played by two players on a finite set of \emph{states}. One of the states is designated as the \emph{initial} state. The two players are referred to as the \emph{maximizer} and the \emph{minimizer}. Each state is controlled by one of the players. From each state there is a finite set of \emph{actions} that can be taken by the player controlling the state. Each action has an immediate \emph{reward} for the maximizer associated with it, which is also the immediate \emph{cost} incurred by the minimizer. Each action also has a probability distribution on states associated with it. The next state of the game is chosen according to this distribution, independent of all previous events.

The two players jointly construct a \emph{trajectory}, also known as a \emph{play}. The play is in general is a stochastic object, even if the players use deterministic strategies, see below, due to the randomized nature of the transitions. A play has an infinite sequence of rewards $r_0,r_1,\ldots$ associated with it. An \emph{aggregator} $\agg:\RR^\omega\to\RR$ is then used to convert this infinite sequence of rewards into a single real number which is considered to be the \emph{outcome} of the play. One of the simplest aggregators is the \emph{discounted} aggregator $\dsc_\lambda(r_0,r_1,\ldots) = (1-\lambda)\sum_{j=0}^\infty \lambda^j r_j$, where $0<\lambda<1$ is the discount factor. Another common aggregator is the \emph{limiting average} aggregator $\underline{\avg}(r_0,r_1,\ldots) = \liminf_{N\to\infty} \frac{1}{N}\sum_{j=0}^{N-1}r_j$.

A \emph{positional} strategy of a player is a choice of a specific action from each state controlled by the player. It is known that, under some minimal technical assumptions, both players have optimal positional strategies. (Due to the infinite nature of the games, this is a nontrivial theorem that requires a non-trivial proof. For more details see Appendix~\ref{A:Games}.)

The natural computational problem associated with 2-TBSG is: Given a 2-TBSG with a specified aggregator, find optimal positional strategies for the two players and the \emph{values} of the game starting from each state. It is a major open problem whether this problem can be solved in polynomial time.

For a positional strategy $\sigma$ of the maximizer, we let $val^{\sigma}(s)$ be the expected outcome of the game that start from state~$s$ when the maximizer uses $\sigma$ and the minimizer uses a \emph{best response} $\tau=\tau(\sigma)$, i.e., an optimal strategy for the minimizer when the maximizer uses~$\sigma$. We usually assume that the set of states is $[n]=\{1,2,\ldots,n\}$, where~$n$ is the number of states. We then let $val^{\sigma}=(val^{\sigma}(1),\ldots,val^{\sigma}(n))$. For two vectors $\bu,\bv\in\RR^n$ we say that $\bu\preceq\bv$  if and only if $u_i\le v_i$, for $i\in[n]$. We say that $\bu\prec\bv$ if and only if $\bu\preceq\bv$ and $u_i<v_i$ for at least one $i\in[n]$.

Suppose that $\sigma$ uses action $a$ from some state~$s$ and that $b$ is another action that can be chosen at~$s$. We let $\sigma[b]$ be the strategy obtained by using~$b$ instead of~$a$ from~$s$. It is known that for every~$\sigma$ and~$b$, either $val^{\sigma}\preceq val^{\sigma[b]}$ or $val^{\sigma[b]}\preceq val^{\sigma}$. For discounted 2-TBSGs, under some natural non-degeneracy assumption, we get that either $val^{\sigma}\prec val^{\sigma[b]}$ or $val^{\sigma[b]}\prec val^{\sigma}$. Furthermore, it is known that if $\sigma$ is not an optimal strategy for the maximizer, then there is at least one action~$b$ such that $val^{\sigma}\prec val^{\sigma[b]}$. For non-discounted 2-TBSGs, things are a bit more complicated. There may not be an action~$b$ for which $val^{\sigma}\prec val^{\sigma[b]}$ even if~$\sigma$ is not an optimal strategy.

To deal with non-discounted 2-TBSGs, we can define an appropriate \emph{extended valuation} $\VAL^{\sigma}$ for which the above properties do hold. (Each component of $\VAL^{\sigma}$ is usually a pair composed of $val^{\sigma}(i)$ and a \emph{potential}, or \emph{bias}, and pairs are ordered lexicographically. The exact definition is nontrivial, see \cite{AkianCDG13} and the other definitions in Appendix~\ref{A:Games}.) For discounted games we can simply take $\VAL^\sigma=val^\sigma$. Also, for any strategy~$\sigma$, $\VAL^\sigma$ can be computed in polynomial time.

This suggests a natural class of \emph{strategy iteration} algorithms for finding an optimal strategy for the maximizer. Start from any initial positional strategy~$\sigma$. For every action~$b$ not currently used by~$\sigma$, check whether $\VAL^{\sigma}\prec \VAL^{\sigma[b]}$. If so, move to $\sigma[b]$. Otherwise, $\sigma$ is an optimal strategy for the maximizer. An optimal strategy for the minimizer can be found in a similar manner. Even better, once an optimal strategy of the maximizer is known, an optimal strategy of the minimizer, which is an optimal strategy in a 1-player game, can actually be found in polynomial time.

Given a non-degenerate 2-TBSG, we can build a corresponding AUSO as follows. Assume that $S_0=[n_0]=\{1,2,\ldots,n_0\}$ is the set of states controlled by the maximizer. Let $A_1,A_2,\ldots,A_{n_0}$ be the set of actions available from each one of these $n_0$ states. A vertex $\ba\in\praa$ naturally corresponds to a positional strategy of the maximizer. If $\bb\in \praa$ is a neighbor of~$\ba$, direct the edge from~$\ba$ to~$\bb$ if $\VAL^{\ba}\prec\VAL^{\bb}$, and from~$\bb$ to~$\ba$ otherwise. The correctness of the strategy iteration algorithm implies that the obtained orientation is indeed an AUSO.

\section{The \RAR\ Algorithm}\label{sec:Random-Action}

Let $\vec{H}\in \vec{\mathcal{H}}_{AUSO}(\pra)$ be an AUSO.
We assume that the sets $A_1,A_2,\ldots,A_n$ are disjoint. 
A vertex $\ba=(a_1,\dots,a_n)\in\vec{H}$ can thus also be viewed as a set $\{a_1,\dots,a_n\}\in\bigcup_{i=1}^n A_i$. Our goal is to find the sink of $\vec{H}$.

We also use the following two convenient notations. If $\vec{H}\in \vec{\mathcal{H}}_{AUSO}(\pra)$ and $r\in \bigcup_{i=1}^n A_i$, we let $\vH\setminus r$ be the AUSO obtained by removing~$r$ from $\vH$. More formally, if $r\in A_i$, we let $A'_i=A_i\setminus\{r\}$ and $A'_j=A_j$, for $j\ne i$. Then, $\vH\setminus r=\vH(\prod_{i=1}^n A'_i)$. Also, if $\ba=(a_1,\dots,a_n)\in \vH$ and $r\in A_i\setminus\{a_i\}$, we let $\ba[r]=\ba\sm \{a_i\} \cup \{r\}$. Alternatively, $\ba[r]=(a'_1,\dots,a'_n)$, where $a'_i=r$ and $a'_j=a_j$, for $j\ne i$, i.e., $\ba[r]$ is obtained by switching the action at state~$i$ to~$r$, leaving all other actions unchanged.

The \RAR\ algorithm recursively constructs a directed path in $\vec{H}$ starting from some initial vertex $\ba\in\vec{H}$ and ending at the unique sink of $\vec{H}$ in the following way:

\begin{enumerate}
    \item If $\ba$ is the only vertex of $\vH$, and hence also its sink, return $\ba$.
    \item Choose a \emph{random action} $r\in (\bigcup_{i=1}^n A_i)\setminus \ba$, i.e., an action \emph{not} used by the current vertex~$\ba$, and \emph{remove} it from $\vH$ to obtain $\vH\setminus r$. 
    (This step gives the algorithm its name.)
    \item Recursively find a directed path in $\vH\setminus r$ leading from $\ba$, which is also a vertex of $\vH\setminus r$, to the sink $\bb$ of $\vH\setminus r$. 
    \item Let $\bc=\bb[r]$ be the vertex in which~$r$ replaces the appropriate action in~$\bb$. Since $\bb$ is the sink of $\vH\setminus r$, all edges incident on~$\bb$ are directed towards~$\bb$, except possibly the edge between $\bb$ and $\bc=\bb[r]$. If the edge from $\bc$ to $\bb$ is also directed towards $\bb$, then $\bb$ is also the sink of $\vH$ and the algorithm is done.
    \item Otherwise, the algorithm adds the edge $\bb\to\bc$ to the path and recursively finds a directed path in $\vH$ from $\bc$ to the sink of $\vH$.
    \end{enumerate}

Pseudocode of the \RAR\ algorithm is given in Figure~\ref{F:RAR-alg} of Appendix~\ref{A:RAR}.
The correctness of the \RAR\ algorithm is immediate. It constructs a directed path in~$\vH$ and only stops when a sink is reached. In principle, however, the algorithm may actually visit all vertices of $\vec{H}$. This happens, however, with tiny probability. Surprisingly, the complexity of \RAR\ is subexponential. 

When called on an AUSO with~$n$ states and~$m$ actions, \RAR\ first makes a recursive call on an AUSO with one action less, i.e., on an AUSO with~$n$ states and $m-1$ actions. It then conditionally makes a second recursive call on the same AUSO, but from a new starting vertex. The intriguing part of the analysis shows that the new starting vertex usually encodes a lot of information that vastly speeds up the second recursive call.

\section{Analyzing the \RAR\ algorithm}\label{sec:Analysis}

If $\vH\in\cH$, we assume that $f=f_{\vH}:\prod_{i=1}^n A_i\to\mathbb{Z}$ is an injective function 
such that $(\ba\to\bb)\in\vH$ implies $f(\ba)<f(\bb)$. Such a function exists as $\vH$ is acyclic. (See also Appendix~\ref{A:AUSOs}.) We stress that~$f$ is only used in the analysis and cannot be used by the algorithm. 
%
The following notion plays a central part in the analysis.

\begin{definition}[Frozen states]
    Let $\vH\in\cH$ and $\ba\in\pra$. State~$i$ is said to be \emph{frozen}, with respect to~$\ba=(a_1,\dots,a_n)\in \pra$, if and only if for every $\bb=(b_1,\dots,b_n)\in\pra$ such that $f(\ba)<f(\bb)$ we have $a_i=b_i$.
\end{definition}

If state~$i$ is frozen, with respect to the current vertex $\ba$, then state~$i$ and all its actions are effectively removed from~$\vH$. The algorithm is nominally running on $\vH$ but is effectively running on $\vH_{-i}\in \vec{\mathcal{H}}_{\text{AUSO}}(\prod_{j\ne i}A_i)$. The only difference between the operation of the algorithm on $\vH$ and on $\vH_{-i}$ is that on $\vH$ the algorithm occasionally checks whether the $i$-th edge of the current vertex is directed outwards and finds out that it is not. This has no effect on the length of the path produced by the algorithm. The same holds if several states are frozen. The intriguing part is that the algorithm usually does not know, and does not need to know, which states are frozen and how many of them there are.
The following definition and lemma explain how frozen states arise.

\newcommand{\drop}{-}

\begin{definition}
    If $\vH\in \cH$ and $b\in \bigcup_{i=1}^n A_i$, let $f^{\drop}(b)=f(\sink(\vH\setminus b))$.
\end{definition}


Recall that if $\bc\in\pra$ is a vertex and $b\in \bigcup_{i=1}^n A_i$ is an action, we say that $b\in\bc$ if and only if~$b$ is used by~$\bc$. (We view~$\bc$ both as a tuple and as a subset of $\bigcup_{i=1}^n A_i$.)

\begin{lemma}\label{L:basic}
    Let  $\vH\in \cH$, $\bc\in\pra$,
    and let $b_1,\dots,b_i\in \bigcup_{i=1}^n A_i$ be distinct actions. 
    If $f^\drop(b_1),\ldots,\allowbreak f^\drop(b_i)<f(\bc)$, then $b_1,b_2,\dots,b_i\in \bc$. Furthermore, each one of  
    $b_1,b_2,\dots,b_i$ belongs to a different state and the states to which they belong are frozen with respect to~$\bc$.
\end{lemma}

\begin{proof}
    If $b_j\notin \bc$ then $f^\drop(b_j)\ge f(\bc)$, a contradiction. If $b_j$ and $b_k$, where $j\ne k$, belong to the same state, then at least one of them, say, $b_j$, is not used by the sink~$\bs$ of~$\vH$. Thus $f^\drop(b_j)=f(\bs)\ge f(\bc)$, again a contradiction. Finally, for every $j\in[i]$, and every $\bc'\in \pra$ such that $f(\bc)<f(\bc')$ we have $f^{\drop}(b_j)<f(\bc)<f(\bc')$ and hence $b_j\in \bc'$. In other words, $b_j$ is frozen with respect to~$\bc$.
\end{proof} 

We assume that $|A_1|,|A_2|,\dots,|A_n|\ge 2$. If $|A_i|=1$, for some $i\in[n]$, then there is no choice to make at~$i$ and the state can be removed. Let $m=\sum_{i=1}^n|A_i|$ be the total number of actions. As we always have $m\ge 2n$, if it convenient to let $m=2n+k$ and perform the analysis in terms of~$n$ and~$k$. At the end we can of course recover the complexity in terms of~$n$ and~$m$. If $k=0$, there are exactly two actions per state and the AUSO is \emph{binary}.

We next define a function $F(n,k)$ and then show that it bounds the expected length of the path constructed by \RARalg\ when run on \emph{any} AUSO with $n$ states and $2n+k$ actions.

\begin{definition}
\label{D:F(n,k)}
    Define a function $F(n,k)$, for $n,k\ge 0$, using the following recurrence relation:
\begin{gather}
   F(n,k) \;=\; F(n,k-1) + \frac{1}{n+k}\sum_{i=1}^n F(n-i,k) \quad,\quad n,k>0 \;, \label{F(n,k)-1} \\
   F(n,0) \;=\; F(n-1,0) + \frac{1}{n}\sum_{i=1}^n F(n-i,0) \quad,\quad n>0 \;, \\
   F(0,k) \;=\; 1 \quad,\quad k\ge 0\;. 
\end{gather}
\end{definition}

Some values of the function $F(n,k)$ for small values of~$n$ and~$k$ are given in Figure~\ref{fig:F-values} of Appendix~\ref{A:F}.

\begin{lemma}
    For every $n,k\ge 0$, the expected length of the path constructed by \RAR\ when run on any AUSO with $n$ states and $2n+k$ actions, starting from any initial vertex, is at most $F(n,k)$.
\end{lemma}

\begin{proof}
    Let $\vH\in \cH$ be the AUSO on which the algorithm is run and let $\ba\in\pra$ be the initial vertex. Let $B = (\bigcup_{i=1}^n A_i)\setminus \ba$ be set of actions not used by~$\ba$. Note that $|B|=n+k$. Let $b_1,b_2,\ldots,b_{n+k}$ be the actions in~$B$ sorted such that
    \[ f^{\drop}(b_1) \le f^{\drop}(b_2) \le \dots \le f^{\drop}(b_n) \le f^{\drop}(b_{n+1}) \le \dots \le f^{\drop}(b_{n+k}) \;.\]
    Let $\bs=\sink(\vH)$. For every $b\not\in \bs$ we have $f^{\drop}(b)=f(\bs)$. Since $|B\sm \bs|\ge k$, we actually have:
    \[ f^{\drop}(b_{n+1}) = \dots = f^{\drop}(b_{n+k})  = f(\bs)\;.\]
    The algorithm chooses a random action $r\in B$ and makes a recursive call on $\vH\sm r$, starting from~$\ba$, producing a directed path from~$\ba$ to~$\bb=\sink(\vH\sm r)$. If $\bb=\sink(\vH\sm r)=\sink(\vH)$, which happens for example if $r\in\{b_{n+1},\ldots,b_{n+k}\}$, then $\bb=\sink(\vH)$ and the algorithm terminates.

    Suppose now that $r=b_i$, for $i\in[n]$. If $\bb=\sink(\vH)$, we are again done. Otherwise, we have $(\bb\to\bc)\in \vH$, where $\bc=\bb[r]$. The algorithm now moves from $\bb$ to~$\bc$. Note that 
    \[ f^{\drop}(b_i) = f(\vH\sm b_i) = f(\bb)<f(\bc) \;.\]
    By the ordering of the $b_i$'s we have:
    \[ f^{\drop}(b_1) \le \dots \le f^{\drop}(b_i) <f(\bc) \;.\]
    Here comes the magic. By Lemma~\ref{L:basic}, we have that $b_1,\ldots,b_i\in \bc$. Furthermore, $b_1,\ldots,b_i$ belong to~$i$ different states and the states to which they belong are now \emph{frozen}. That means that the second recursive call of the algorithm is effectively running on an AUSO with only $n-i$ states. 

    We now argue that the expect length of the path produced by \RAR\ on an AUSO on~$n$ states and $2n+k$ actions, starting from any starting vertex, is bounded by $F(n,k)$. The proof is by induction on $n$ and $k$. If $n=0$, then algorithm returns immediately with a path of length~$1$.
    
    Assume now that $n,k>0$. For every $i\in[n+k]$ the probability that $r=b_i$ is exactly $\frac{1}{n+k}$. (The order we gave to the actions in~$B$ is not known to the algorithm, but this is not needed.) By induction, the expected length of the path constructed by the recursive call on $\vH\sm r$ is at most $F(n,k-1)$. If $i>n$, the algorithm terminates. Otherwise, the algorithm moves to $\bc=\bb[r]$ and starts a second recursive call from there. As argued, $\vH$ now has at least $i$ states frozen with respect to~$\bc$. These states and their actions are effectively removed from $\vH$. (This is the magic.) Thus, the second recursive call is effectively on an AUSO with $n-i$ states and $(2n+k)-2i=2(n-i)+k$ actions. (Note that by removing $i$ states we also remove at least $2i$ actions.) Thus, by the induction hypothesis, the length of the path created by this second recursive call is at most $F(n-i,k)$. The expected total length of the path constructed is thus at most
    $F(n,k-1) + \frac{1}{n+k}\sum_{i=0}^n F(n-i,k)$ which is exactly $F(n,k)$.

    We are left with the case $k=0$. This corresponds to the case of \emph{binary} AUSOs, i.e., AUSOs in which $|A_i|=2$, for every $i\in[n]$. The argument of the previous paragraph still applies, with one important change. The first recursive call removes an action from $\vH$. In the binary case, the state to which the removed action belongs becomes degenerate, i.e., it is left with only one action and this state can effectively be removed from the AUSO. (In particular, this state is frozen.) Thus, the first recursive call is on an AUSO with $n-1$ states and $2(n-1)$ actions. By induction, the expected length of the path returned by the first recursive call is at most $F(n-1,0)$. This gives the modified recurrence relation for the $k=0$ case. 
\end{proof}

We are therefore left with the (nontrivial) task of analyzing the recurrence relation defining $F(n,k)$. An interesting equivalent version of the recurrence relation defining $F(n,k)$ 
is given in Appendix~\ref{A:equiv}.

\section{Solving the recurrence for $F(n,0)$}\label{sec:Binary}

For brevity, we let $F(n)=F(n,0)$. Note that $F(n)$ is an upper bound on the length of the path constructed by the \RAR\ algorithm when called on a binary AUSOs with $n$ states. The function $F(n)$ satisfies the following recurrence relation.
\begin{gather*}
 F(n) \;\le\; F(n-1) + \frac{1}{n}\sum_{i=1}^n F(n-i) \quad,\quad n>0 \;, \\
 F(0) \EQ 1 \;.
\end{gather*}
G{\"a}rtner \cite{Gartner02} obtained the following lemma, easily proved by induction:

\begin{lemma}[G{\"a}rtner \cite{Gartner02}]\label{L:F}
    $\displaystyle \quad F(n) \EQ \sum_{r=0}^n \binom{n}{r}\frac1{r!} \LT \sum_{r=0}^n \frac{n^r}{(r!)^2} \LT \left(\sum_{r=0}^n \frac{n^{r/2}}{r!}\right)^2 \LT e^{2\sqrt{n}}\;.$
\end{lemma}

Lifschitz and Pittel \cite{LiPi81} showed that 
$F(n) = \frac{e^{2\sqrt n}}{2\sqrt{\pi e}\,n^{1/4}}(1+O(n^{-1/2}))$. (See Appendix~\ref{A:Binary}.)

\section{A simple upper bound on $F(n,k)$}\label{S:Simple}

The upper bound $F(n)\le e^{2\sqrt{n}}$ can be used to obtain the following simple, non-trivial, and sometimes asymptotically tight, upper bound on $F(n,k)$. The simple proof of the Lemma is given in Appendix~\ref{A:Simple}.

\begin{lemma}\label{L:n+k}
  For every $n,k\ge 0$ we have: $\quad F(n,k) \LE F(n+k) \LT e^{2\sqrt{n+k}} \;.$
\end{lemma}

We believe that for $0\le k\le n$, the bound $F(n,k)\le e^{2\sqrt{n+k}}$ is asymptotically tight, up to polynomial factors, but it becomes extremely loose, as we shall see, when $k>n$.

\section{A tighter upper bound on $F(n,k)$, when $k\ge n$}\label{S:Recurrence}

We now present an improved analysis that yields a tighter upper bound on $F(n,k)$ when $k>n$.

We define an explicit function $U(n,k)$ and show by induction that $F(n,k)\le U(n,k)$. For the induction to work, $U(n,k)$ should capture the correct asymptotic behavior of $F(n,k)$, at least at the exponential scale. The definition of $U(n,k)$ is far from trivial or obvious. In Appendix~\ref{app:U} we explain how this form of~$U(n,k)$ was derived. We define:
\begin{equation}
    U(n,k) \EQ (n+k)\, e^{2\sqrt{n\Phi(\frac{k}{n})}} \;,
\end{equation}
where $\Phi=\Phi(\alpha)$ is a solution of the following differential equation
\begin{equation}
    \alpha (\Phi')^2 - \Phi \Phi' + \frac{\Phi}{1+\alpha} \EQ 0 \;, \label{eq:Phi}
\end{equation}
together with the initial condition $\Phi(1)=2$. Solving the quadratic equation for $\Phi'=\Phi'(\alpha)$ we get
\begin{equation}
    \Phi' \EQ \frac{\Phi \pm \sqrt{\Phi^2 - \frac{4\alpha}{1+\alpha}\Phi}}{2\alpha} \;. \label{eq:Phi2}
\end{equation}
Thus, the differential equation~\eqref{eq:Phi} actually has two different solutions $\Phi_+(\alpha)$ and $\Phi_-(\alpha)$ obtained by adding or subtracting the square root in equation~\eqref{eq:Phi2}. (To show that these two solutions are well-defined, we need to show that $\Phi\ge \frac{4\alpha}{1+\alpha}$ so that the square root produces a real number.) To obtain a tighter upper bound we naturally choose $\Phi(\alpha)=\Phi_-(\alpha)$.

For $0\le \alpha\le 1$, we have $\Phi_-(\alpha)=1+\alpha$, as can be easily checked by differentiating. At $\alpha=1$, the discriminant, i.e., the square root, becomes~$0$. Also note that $\Phi_-(1)=2$. For $\alpha>1$ we actually have $\Phi_+(\alpha)=1+\alpha$ while $\Phi_-(\alpha)$ grows much more slowly. There does not seem to be a closed-form expression for $\Phi(\alpha)=\Phi_-(\alpha)$, when $\alpha>1$, but we can compute it numerically by integrating the differential equation. We can also prove rigorous lower and upper bounds on $\Phi(\alpha)$ and analyze its asymptotic behavior. In Appendix~\ref{app:Phi} we prove the following technical lemma that show that $\Phi(\alpha)=\Phi_-(\alpha)$ is well-defined and that its growth rate is essentially logarithmic, rather than linear, and that it is strictly concave for $\alpha> 1$.

\begin{lemma}\label{L-ul-Phi}
    The function $\Phi(\alpha)=\Phi_-(\alpha)$ is well-defined for every $\alpha\ge 0$. For $0\le \alpha\le 1$, we have $\Phi(\alpha)=1+\alpha$. For every $\alpha> 1$, $\Phi(\alpha)$ is strictly concave, i.e., $\Phi''(\alpha)<0$, and 
    \[ 2+\s+\ln(\s+1)-\left(\frac{\s}{\s+1}\right)^3 \LE \Phi(\alpha) \LE 2+\s+\ln(\s+1) \quad,\quad \s=\ln\frac{1+\alpha}{2}\;.\]
\end{lemma}

Plugging $\Phi(\alpha)=1+\alpha$ for $0\le\alpha\le 1$ gives us an upper bound $F(n,k)\le U(n,k)=(n+k)\,e^{2\sqrt{n+k}}$ which is slightly worse than what we already have. However, when $\alpha>1$ the upper bound obtained becomes much better than the upper bound obtained in the previous section.

Our goal is proving the following theorem, which implies the main result of the paper:

\begin{theorem}\label{T-upper}
    For every $n,k>0$ we have $F(n,k) \LE U(n,k) \EQ (n+k)\, e^{2\sqrt{n\Phi(\frac{k}{n})}}$.
\end{theorem} 

The proof of Theorem~\ref{T-upper} is by induction on~$n$ and~$k$. For $k\le n$ we already know that the inequality holds, even without the multiplicative $n+k$ factor. This is the basis of the induction. The $n+k$ factor simplifies certain steps of the proof, as we shall see below. 

Throughout the proof we let $\alpha=\frac{k}{n}$ and 
\begin{equation*}
    S(n,k) \EQ 2\sqrt{n\Phi\biggl(\frac{k}{n}\biggr)} \EQ 2\sqrt{n\Phi(\alpha)} \quad,\quad
U(n,k) \EQ (n+k)\,e^{S(n,k)} \;. 
\end{equation*}
In what follows, we use the abbreviations $\Phi=\Phi(\alpha)$, $\Phi'=\Phi'(\alpha)$ and $\Phi''=\Phi''(\alpha)$, where $\Phi'(\alpha)$ and~$\Phi''(\alpha)$ are the first and second derivatives of $\Phi(\alpha)$ with respect to $\alpha$.
We let 
\[ S_n=\frac{\partial S}{\partial n} \quad,\quad S_k=\frac{\partial S}{\partial k} \quad,\quad S_{nn}=\frac{\partial^2 S}{\partial n^2} \quad,\quad S_{kk}=\frac{\partial^2 S}{\partial k^2} \] be the corresponding partial derivatives of $S=S(n,k)$ with respect to~$n$ and~$k$. (We do not need~$\frac{\partial^2 S}{\partial n\partial k}$ in our analysis.) We always assume that they are evaluated at the point $(n,k)$. The following expressions for the partial derivatives are easily obtained using the chain rule of differentiation. The proofs are given in Appendix~\ref{A:partial}.

\begin{lemma}\label{L:SnSk}
With $S(n,k) = 2\sqrt{n\Phi(\frac{k}{n})}$ and $\alpha=\frac{k}{n}$, we have:
    \[ S_n \EQ \frac{\Phi - \alpha \Phi'}{\sqrt{n \Phi}} \quad,\quad S_k \EQ \frac{\Phi'}{\sqrt{n \Phi}} 
     \quad,\quad
          S_{nn} \EQ \frac{\alpha^2 \Phi \Phi'' - \frac{1}{2}(\Phi - \alpha\Phi')^2}{(n\Phi)^{3/2}} \quad,\quad
    S_{kk} \EQ \frac{\Phi \Phi'' - \frac{1}{2}(\Phi')^2}{(n\Phi)^{3/2}}  \;.
\]
\end{lemma}

The following lemma is central for our analysis.
\begin{lemma}\label{L:PDE}
    The function $S(n,k)=2\sqrt{n\Phi(\frac{k}{n})}$ satisfies the following partial differential equation:
    \[ (n+k)S_nS_k \EQ 1\;.\]
\end{lemma}

\begin{proof}
    Using Lemma~\ref{L:SnSk} and the differential equation \eqref{eq:Phi} defining $\Phi(\alpha)$ we get:
    \[
    (n+k)S_nS_k \EQ (1+\alpha)n \, \frac{\Phi - \alpha \Phi'}{\sqrt{n \Phi}}  \,\frac{\Phi'}{\sqrt{n \Phi}}
        \EQ (1+\alpha)n \,\,\frac{\frac{\Phi}{1+\alpha}}{n \Phi}\EQ 1 \;. \qedhere
    \]
\end{proof}

Finally, we also need the following technical lemma whose proof is given in Appendix~\ref{A:concave}.

\begin{lemma}\label{L-S-concave}
    The function $S(n,k)$ is strictly concave in~$n$ and in~$k$ when $0<n<k$. In other words, $S_{nn}<0$ and $S_{kk}<0$ when $0<n<k$.
\end{lemma}

With these preparations, we are now ready to prove Theorem~\ref{T-upper}.

\begin{proof}[Proof of Theorem~\ref{T-upper}]
Our goal is to prove that
\begin{equation}
    U(n,k) \GE U(n,k-1) + \frac{1}{n+k}\sum_{j=0}^{n-1}U(j,k) \quad,\quad 0<n<k \;.
    \label{eq:up}
\end{equation} 
Using induction on $n+k$ we then have
\begin{equation*}
    F(n,k) \EQ F(n,k-1) + \frac{1}{n+k}\sum_{j=0}^{n-1} F(j,k)
           \LE U(n,k-1) + \frac{1}{n+k}\sum_{j=0}^{n-1} U(j,k) \LE U(n,k) \;.
\end{equation*}

All that remains, therefore, is to prove inequality~\eqref{eq:up}, which is equivalent to the following inequality:
\begin{equation}
   \frac{U(n,k-1)}{U(n,k)} + \frac{1}{(n+k)\,U(n,k)}\sum_{j=0}^{n-1}U(j,k) \LE 1 \quad,\quad 0<n<k \;.
   \label{eq:Q}
\end{equation}

We begin by bounding the first term of \eqref{eq:Q}. Let $N=n+k$ and $S_k=\frac{\partial S}{\partial k}$. By the mean value theorem we have $S(n,k)-S(n,k-1)=S_k(n,y)$, for some $y\in[k-1,k]$. By Lemma~\ref{L-S-concave} the function $S(n,k)$ is concave in~$k$ and hence $S_k(n,y)$ is decreasing in~$y$. Hence $S(n,k)-S(n,k-1)\ge S_k(n,k)=S_k$, or equivalently $S(n,k-1)-S(n,k)\le -S_k$. Thus
\begin{equation}
    \frac{U(n,k-1)}{U(n,k)} \EQ  \frac{n+k-1}{n+k}\,e^{S(n,k-1)-S(n,k)} \LE (1-\frac{1}{N})\,e^{-S_k} \;. \label{eq:U1}
\end{equation}

We next bound the sum appearing in the second term of \eqref{eq:Q}. Let $S_n=\frac{\partial S}{\partial n}(n,k)$. By Lemma~\ref{L-S-concave} the function $S(n,k)$ is also concave in~$n$. Thus, the function $S(x,k)$ lies below its tangent at $x=n$, i.e.,  
\[ S(x,k) \LE S(n,k) + S_n\cdot(x-n) \;.\]
Thus,
\[ \frac{1}{n+k}\sum_{j=0}^{n-1} U(j,k) \LE \int_0^n e^{S(x,k)}\,dx \LE e^{S(n,k)} \int_0^n e^{S_n(x-n)\,dx} \EQ \frac{e^{S(n,k)}}{S_n}\left[e^{S_n(x-n)}\right]_0^n \LE \frac{e^{S(n,k)}}{S_n} \;. \]

It follows that
\begin{equation}
    \frac{1}{(n+k)\,U(n,k)}\sum_{j=0}^{n-1}U(j,k) \LE \frac{1}{Ne^{S(n,k)}}\frac{e^{S(n,k)}}{S_n} \EQ \frac{1}{NS_n} \EQ S_k \;,
    \label{eq:U2}
\end{equation}
where the last step uses Lemma~\ref{L:SnSk} that states that $(n+k)S_nS_k=1$.

Using inequalities \eqref{eq:U1} and~\eqref{eq:U2}, we get that to prove inequality~\eqref{eq:Q}, it is enough to prove that
\begin{equation}
    (1-\frac{1}{N})e^{-S_k} + S_k \LE 1 \label{eq:Sk}\;.
\end{equation}

It can be seen here that the $(1-\frac{1}{N})$ factor, which comes from the multiplicative $(n+k)$ factor in the definition of $U(n,k)$, is essential for our proof plan, as without it the inequality $e^{-S_k} + S_k\le 1$ does not hold, even for very small values of~$S_k$. With the help of the $(1-\frac{1}{N})$ term, we are able to show that \eqref{eq:Sk} holds for all values that $S_k$ can assume.

Using the upper bound $e^{-x}\le 1-x+\frac{x^2}{2}$, that holds for every $x\ge 0$, we get
\begin{equation*}
(1-\frac{1}{N})e^{-S_k} + S_k \LE (1-\frac{1}{N})(1-S_k+\frac{1}{2}S_k^2) + S_k 
\EQ (1-\frac{1}{N}) + \frac{S_k}{N} + \frac12(1-\frac{1}{N})S_k^2\;. 
\end{equation*}
Thus, to prove inequality~\eqref{eq:Sk} is it enough to show that 
\begin{equation}
-\frac{1}{N} + \frac{S_k}{N} + \frac12(1-\frac{1}{N})S_k^2 \LE 0\;. \label{eq:quad2}
\end{equation}
The left-hand side is a quadratic expression in~$S_k$ and it is increasing for $S_k\ge 0$. Evaluating this quadratic expression at $S_k=\frac{1}{\sqrt{N}}$ we get a value of $\frac{2\sqrt{N}-N-1}{2N^2}$ which is negative for every $N\ge 2$.
Thus, to establish inequality~\eqref{eq:quad2} and finish the proof it is enough to show that $S_k\le \frac{1}{\sqrt{N}}$.

From Lemma~\ref{L:SnSk} and the differential equation \eqref{eq:Phi2} we have:
\[ S_k \EQ \frac{\Phi'}{\sqrt{n\Phi}} \EQ \frac{\Phi-\sqrt{\Phi^2-\frac{4\alpha}{1+\alpha}\Phi}}{\sqrt{n\Phi}} 
\EQ \frac{1}{2\alpha\sqrt{n}}\left(\sqrt{\Phi}-\sqrt{\Phi-\frac{4\alpha}{1+\alpha}}\right) \;.\]
The function $\sqrt{x}-\sqrt{x-c}$ is decreasing in~$x$ for every $x\ge c$. Thus, $\sqrt{x}-\sqrt{x-c}\le\sqrt{c}$ for every $x\ge c$. As $\Phi\ge \frac{4\alpha}{1+\alpha}$, we get that
\[ S_k \LE \frac{1}{2\alpha\sqrt{n}} \sqrt{\frac{4\alpha}{1+\alpha}} \EQ \frac{1}{\sqrt{\alpha(1+\alpha)n}} \EQ \frac{1}{\sqrt{\alpha N}} \LE \frac{1}{\sqrt{N}} \;,\]
since $\alpha\ge 1$. This completes the proof of the Theorem.
\end{proof}

\section{Concluding results}\label{sec:concl}

We obtained an improved and streamlined analysis of the \RAR\ algorithm for solving 2-player Turn-Based Stochastic Games (2-TBSGs) and for finding sinks of non-binary AUSOs. The analysis uses for the first time the natural hypercube structure of the games that was previously ignored in the non-binary case. This leads to a nicer recurrence relation closely tied to very simple and natural partial differential equations. We use a discrete version of the WKB technique to obtain a tight upper bound on the solution of the recurrence relation. As a result, the upper bound on the complexity of the \RAR\ algorithm is improved from $e^{O\bigl(\sqrt{n\ln\frac{m}{\sqrt n}}\,\bigr)}$ to $e^{O\bigl(\sqrt{n\ln\frac{m}{n}}\,\bigr)}$. The improvement is especially significant when $m=O(n)$, as then the upper bound is improved from $e^{O(\sqrt{n\ln n})}$ to $e^{O(\sqrt{n})}$, implying a smooth transition between binary and non-binary games.

Although the overall improvement obtained is fairly small, we think that it is significant given the importance of 2-TBSGs and given that we believe that it gives the true asymptotic behavior of the algorithm. As mentioned, the \RAR\ algorithm is the fastest known general algorithm for solving such games.

Major open problems remain of course. The most important open problem is probably obtaining even faster algorithms for 2-TBSGs or for some interesting subfamilies of games such as Mean Payoff Games (MPGs), as was done about a decade ago for Parity Games (PGs). Another interesting open problem is obtaining deterministic subexponential time algorithms for solving 2-TBSGs. The fastest known determinitic algorithms currently have an exponential complexity even for binary games.

\subsection*{AI Disclosure} We used Google's Gemini and OpenAI's ChatGPT to assist with the solution of the recurrence relation obtained for $F(n,k)$. The tool materially affected Section~\ref{S:Recurrence} and Appendices~\ref{app:U} to~\ref{A:concave}. The authors verified the correctness and originality of all content including references.

\newpage
\clearpage
\appendix
\addtolength{\headsep}{0.75cm}   
\addtolength{\textheight}{-1.5cm} 
\pagestyle{appendix}

\section{Context and significance}\label{A:context}

The \RAR\ algorithm is an adaptation of a dual randomized simplex algorithm of Matou{\v{s}}ek, Sharir and Welzl \cite{MaShWe96}. Kalai \cite{Kalai92,Kalai92b} independently developed a primal randomized simplex algorithm that is dual version of the algorithm of Matou{\v{s}}ek, Sharir and Welzl \cite{MaShWe96}. The relation between the two algorithms was noted by Goldwasser \cite{Goldwasser95}. 

Matou{\v{s}}ek, Sharir and Welzl \cite{MaShWe96} obtain an $e^{O(\sqrt{n\ln(m/\sqrt n)}\,)}$ upper bound on the running time of their dual simplex algorithm, where~$n$ is the number of variables of the input linear program and~$m$ is the number of inequality constraints. Their analysis is fairly complicated and employs generating functions techniques. For more on their analysis, see G{\"{a}}rtner \cite{Gartner95PhD}.

Kalai \cite{Kalai92,Kalai92b}, G{\"{a}}rtner \cite{Gartner95} and Hansen and Zwick \cite{HaZw15} obtained versions of the primal Random Facet algorithm with an improved upper bound of $e^{O(\sqrt{n\ln(m/\sqrt n)}\,)}$. These algorithms can also be adapted to games and AUSOs and as mentioned provide an alternative way of obtaining an $e^{O(\sqrt n)}$ bound for games with $m=O(n)$, but using much more complicated algorithms and with worse constant factors in the exponent.

It is a major open problem whether there is a polynomial version of the simplex algorithm or, more generally, a strongly polynomial time algorithm for solving linear programming problems.

2-player, 0-sum, Turn-Based Stochastic Games (2-TBSG) are an intriguing class of games that form a generalization of many infinite duration games such as Simple Stochastic Games (SSGs) \cite{Condon92,AnMi09}, Mean Payoff Games (MPGs) \cite{EhMy79,GuKaKh88,ZwPa96}, Energy Games (EGs) \cite{chakrabarti2003resource,bouyer2008infinite} and Parity Games (PGs) \cite{EmJu91,emerson1993model,mcnaughton1993infinite}. For a thorough treatment of all these and other games, see Fijalkow et al.~\cite{gamesbook}. The problems of solving these games belong to $\text{NP}\cap \text{co-NP}$, yet it is a major open problem whether any of these games can be solved in polynomial time. This is one of the simplest and most natural combinatorial problem with this status.

Parity Games (PGs) can be solved in $n^{O(\log n)}$ time (see \cite{CaludeJKLS22} and \cite{CzerwinskiDFJLP19,FearnleyJKSSW19,GimbertI17,JurdzinskiL17,Lehtinen2018modal,LehtinenPSW22}), but these improved bounds do not seem to extend even to MPGs \cite{FijalkowGO20,ColcombetFGO22}. Mean Payoff Games (MPGs) and Energy Games (EGs) can be solved in pseudo-polynomial time \cite{ZwPa96,brim2011faster,CominR17}. Discounted versions of all these games can be solved in strongly polynomial time for any fixed discount factor  \cite{Ye11a,HaMiZw13}. (In a discounted game, the reward obtained in the $i$-th step is multiplied by $\lambda^i$.)

1-player stochastic games, also known as Markov Decision Processes (MDPs) \cite{bellman1957markovian,Howard60,Derman72,Puterman94}, can be solved in polynomial time using linear programming algorithms. Discounted MDPs with a fixed discount factor $\lambda<1$ can be solved in strongly time using either a simplex algorithm \cite{Ye11a,Scherrer16}, or an interior point algorithm \cite{Ye05}. It is a major open problem whether non-discounted MDPs can be solved in strongly polynomial time. This is one of the simplest and most natural families of linear programs for which this problem is open.

Concurrent 2-player Stochastic Games, introduced by Shapley \cite{Shapley53}, form a huge generalization of 2-TBSG. For a thorough mathematical treatment of such games see Filar and Vrieze \cite{FiVr12} or Solan \cite{solan2022course}. Concurrent Stochastic games of 2 or more players are also known as \emph{Markov Games}. This is a huge area of research that spans Reinforcement Learning (RL) \cite{sutton2018reinforcement} and Algorithm Game Theory.

The fastest known algorithm for the non-discounted 2-TBSGs, or discounted 2-TBSGs with a discount factor~$\lambda$ very close to~$1$, is the \RAR\ algorithm considered in this paper. We obtain a tighter analysis of this algorithm that goes for the first time beyond the analysis of Matou{\v{s}}ek, Sharir and Welzl \cite{MaShWe96} that applies in a more general setting. This analysis also provides a uniform treatment of all games, providing a smooth transition between binary and non-binary game. We also think that the application of the discrete WKB technique to problems arising in the analysis of algorithms may be interesting in its own right.

2-TBSG are interesting not only from the computational complexity perspective. They are motivated by various problems in automata, logic and automatic program verification and synthesis. For more details, see Gr{\"a}del, Thomas, and Wilke \cite{GrThWi02}, Baier and Katoen \cite{baier2008principles}, and Apt and Gr{\"a}del \cite{ApGr11} and the huge number of references contained in them.

Akian, Gaubert and Guterman \cite{AkGaGu12} obtained an equivalence between tropical polyhedral feasibility problems and MPGs. Allamigeon, Gaubert and Skomra \cite{AGS18} showed that solving generic nonarchimedean semidefinite programs are equivalent to solving Stochastic Games. Allamigeon, Benchimol, Gaubert and Joswig \cite{ABGJ14} show that combinatorial simplex algorithms can solve MPGs, showing another intriguing connection to Linear Programming algorithms.

\section{Acyclic Unique Sink Orientations (AUSOs)}\label{A:AUSOs}

We begin with standard definitions of hypercubes, oriented hypercubes, sub-hypercubes, and sinks of oriented hypercubes.

\begin{definition}[Hypercubes]
Let $A_1,A_2,\ldots,A_n$ be finite sets. The \emph{hypercube}, or just \emph{cube}, $H(\pra)=(V,E)$
is a graph whose vertex set is $V=\pra$ and its edge set is $E=\{ \{\ba,\bb\} \mid \ba,\bb\in V \text{ and } d_H(\ba,\bb)=1\}$, where $d_H(\ba,\bb)=|\{i\in[n]\mid \ba_i\ne\bb_i\}|$ is the Hamming distance between~$\ba$ and~$\bb$. In other words, two vertices $\ba,\bb\in\pra$ are adjacent if and only if they differ in exactly one coordinate. For brevity, we often write $\{\ba,\bb\}\in H(\pra)$ instead of $\{\ba,\bb\}\in E$.
\end{definition}

\begin{definition}[Oriented cubes]
   An \emph{oriented cube} $\vec{H}$ on $\pra$ is obtained by assigning each edge of $H(\pra)$ a direction. More formally, $\vec{H}$ is a directed graph $(V,\vec{E})$, where $V=\pra$ and $\vec{E}$ is such that for every $\{\ba,\bb\}\in H(\pra)$, either $(\ba,\bb)\in \vec{E}$ or $(\bb,\ba)\in \vec{E}$, but not both. We often write $(\ba\to\bb)\in \vec{H}$ instead of $(\ba,\bb)\in \vec{E}$ to stress that the edge is viewed as directed from~$\ba$ to~$\bb$. We let $\oH$ be the set of all oriented cubes on $\pra$.
\end{definition}

We refer to the indices $i\in[n]$ as \emph{coordinates}, and to a set~$A_i$ as the set of \emph{choices} available in the $i$-th coordinate. Due to the relation of AUSOs to the games that we introduce in the next section, we also refer to the indices $i\in [n]$ as \emph{states}, and to a set~$A_i$ as the set of \emph{actions} available at state~$i$. We use these terms interchangeably. We also refer to $\ba\in\pra$ as a \emph{vertex}, or as a \emph{strategy}, in the terminology of games.

\begin{definition}[Acyclicity]
    An oriented cube $\vec{H}\in\oH$ is \emph{acyclic} if it contains no directed cycles.
\end{definition}

\begin{definition}[Subcubes]
    If $\vec{H}\in\oH$ and $A'_1\subseteq A_1,\dots,A'_n\subseteq A_n$, then the \emph{subcube} $\vec{H}[\prod_{i=1}^n A'_i]$ is the oriented cube on $\prod_{i=1}^n A'_i$ induced by~$\vec{H}$.
\end{definition}

\begin{definition}[Sinks]
    Let $\vec{H}\in\oH$. A vertex $\ba$ of $\vec{H}$ is said to be a \emph{sink} of~$\vH$ if and only if all edges incident to~$\ba$ are directed towards~$\ba$.
\end{definition}

We next define the central notion of Acyclic Unique Sink Orientations (AUSOs) of hypercubes. Our definition generalizes the definition of binary AUSOs given by Szab{\'o} and Welzl \cite{SzWe01}.

\begin{definition}[Acyclic Unique Sink Orientations (AUSO)]
    An oriented cube $\vH\in\oH$ is an \emph{Acyclic Unique Sink Orientation (AUSO)} of $\pra$ if and only if $\vH$ is \emph{acyclic} and every subcube of $\vH$, including itself, has a \emph{unique} sink. If $\vH$ is an AUSO, we let $\sink(\vH)$ denote its unique sink. We also let $\cH$ be the set of all AUSOs on $\pra$.
\end{definition}

Every injective function $f:\prod_{i=1}^n A_i\to \mathbb{Z}$ gives rise to an acyclic orientation $\vH_f$ of $\pra$ in the following way:  $(\ba\to\bb)\in \vH_f$ if and only if $f(\ba)<f(\bb)$. Every sink of $\vH_f$ corresponds to a \emph{local maximum} of~$f$. If~$f$ has no local maximum other than the \emph{global maximum}, then $\vH$ has a unique sink. However, $\vH$ is not necessarily an AUSO. $\vH$ is an AUSO if and only if $f$ satisfies the stronger property that every subcube of $H(\pra)$ has no local maximum other than the global maximum on this subcube.

Conversely, if $\vH$ is acyclic, then there exists an injective function $f:\prod_{i=1}^n A_i\to \mathbb{Z}$ such that $(\ba\to\bb)\in \vH_f$ if and only if $f(\ba)<f(\bb)$; simply take a topological ordering of~$\vH$. In the analysis of certain sink-finding algorithms it is sometimes convenient to refer to such a function~$f$. The algorithms themselves have no access to~$f$. (The question whether having access to such a function~$f$ can help find the unique sink is an interesting open problem.) 

A sink-finding algorithm is given \emph{oracle} access to an input AUSO $\vH$. The algorithm can query a vertex $\ba\in\vH$ and receive the direction of all the edges incident on~$\ba$. The complexity of the algorithm is the (expected) number of queries performed by the algorithm before the sink of $\vH$ is identified. It is a major open problem whether there sink-finding algorithm that performs a polynomial, in~$n$, number of queries. The fastest known sink-finding algorithm is the \RAR\ algorithm analyzed in this paper.

Unique Sink Orientations (USOs), also defined by Szab{\'o} and Welzl \cite{SzWe01}, are orientations that have a unique sink in every subcube, but are not required to be acyclic. Finding the sinks of USOs seems to be a harder problem than finding the sinks of AUSOs. In particular, the \RAR\ algorithm cannot be used to find sinks of USOs. G{\"a}rtner and Schurr \cite{GaSc06} showed that a polynomial time algorithm for finding the sink of USOs would yield a strongly polynomial time algorithm for linear programming.

\section{2-player, 0-sum, Turn-Based, Stochastic Games (2-TBSGs)}\label{A:Games}

In this section we give a brief introduction to 2-player, 0-sum, Turn-Based, Stochastic Games (2-TBSG) \cite{HaMiZw13}. Simple Stochastic Games (SSGs) \cite{Condon92,AnMi09} are special cases of such games. Mean Payoff Games (MPGs) \cite{EhMy79,GuKaKh88,ZwPa96} are essentially deterministic versions of 2-TBSGs and are thus also included as special cases. A classical example of a 2-TBSG is \emph{backgammon}, through the set of states is of exponential size. Similarly, a classical example of a deterministic 2-player turn-based game is \emph{chess}, again with an exponential set of states.

As mentioned, the \RAR\ algorithm is the fastest known general algorithm for solving these families of games, unless the games are discounted or the rewards have a special structure. For an in-depth study of these, and other games, see \cite{gamesbook}. In addition to defining 2-TBSGs, we also describe a family of algorithms, known as the \emph{Strategy Iteration} algorithms. The \RAR\ algorithm belongs to this family.

\subsection{Definitions}\label{sub:game-def}

\begin{definition}[2-TBSG]
    A 2-player 0-sum Turn-Based Stochastic Game (2-TBSG) is a tuple $G=(S_0,S_1,A,r,s_0)$. Here~$S_0$ and~$S_1$ are disjoint finite sets of \emph{states}, where the states of~$S_k$, for $k\in\{0,1\}$ are controlled by player~$k$. Player~$0$ is the \emph{maximizer} while player~$1$ is the \emph{minimizer}. We let $S=S_0\cup S_1$. $A$ is a finite set of \emph{actions}. Each action $a\in A$ is a pair $(s,p)$, where $s\in S$ and $p\in\Delta(S)$ is a probability distribution over $S$. We usually assume that $S=[n]=\{1,2,\ldots,n\}$ and let $A_i=\{(i,p) \mid (i,p)\in A\}$, for $i\in[n]$, be the set of actions associated with state $i\in[n]=S$. Note that the sets $A_1,\ldots,A_n$ are disjoint and that $\cup_{i=1}^n A_i=A$.
$r:A\to\RR$ is a function that assigns a \emph{reward} of each action. Finally, $s_0\in S$ is the \emph{initial} state.
\end{definition}

\begin{definition}[Plays]
The two players jointly construct an infinite sequence $(s_0,a_0),(s_1,a_1),\ldots$ of state-action pairs known as a \emph{play}, or \emph{trajectory}, from the initial state~$s_0$ as follows. If $s_j\in S_k$, where $k\in\{0,1\}$, then player~$k$ chooses an action $a_j=(s_j,p_j)\in A$. (Note that~$a_j$ must be an action of from~$s_j$. We assume that each state has at least one action belonging to it.) The next state~$s_{j+1}$ is then chosen according to the probability distribution $p_j\in\Delta(S)$. In general, the sequence $(s_0,a_0),(s_1,a_1),\ldots$ is stochastic, i.e., $s_j$ and $a_j$ for $j>0$ are random variables even if the two players are using deterministic strategies, as defined below. 
\end{definition}

\begin{definition}[Deterministic games]
    A game is said to be \emph{deterministic}, or \emph{non-stochastic}, if for each action $(s,p)\in A$, there is a unique $s'\in S$ such that $p(s')=1$. The action $a=(s,p)$ is then equivalent to a directed edge $s\to s'$ and the game is then simply a directed graph. In both stochastic and non-stochastic games, a state~$s$ is said to be a \emph{stopping} state, if there is only one action $a=(s,p)$ associated with~$s$ and $p(s)=1$.
\end{definition}

\begin{definition}[Sequences of rewards]
 The state-action sequence $(s_0,a_0),(s_1,a_1),\ldots$ gives rise to an infinite sequence of rewards $r_j=r(a_j)$, for $j\ge 0$. As the game is 0-sum, the reward~$r_j$ received by player~$0$ is also the \emph{cost} incurred by player~$1$. Player~$0$, the maximizer, wants to maximize the rewards she receives, while player~$1$, the minimizer, wants to minimize the costs he incurs.
\end{definition}

\begin{definition}[Aggregators]
    To finish the definition of a game, we need to specify how the infinite sequence of rewards is converted into a single number that gives the \emph{outcome} of the game. This is done by a function $\agg:\RR\omega\to \RR$ called an \emph{aggregator}.
    \emph{Discounted} games use the aggregator $\dsc_\lambda(r_0,r_1,\ldots) = (1-\lambda)\sum_{j=0}^\infty \lambda^j r_j$, where $0<\lambda<1$ is the \emph{discount factor}. \emph{Non-discounted} game use the \emph{limiting average} aggregator $\underline{\avg}(r_0,r_1,\ldots) = \liminf_{N\to\infty} \frac{1}{N}\sum_{j=0}^{N-1}r_j$. (Alternatively, the aggregator $\overline{\avg}(r_0,r_1,\ldots) = \limsup_{N\to\infty} \frac{1}{N}\sum_{j=0}^{N-1}r_j$ can be used. In most cases, this makes no difference. Note that $\avg(r_0,r_1,\ldots) = \lim_{N\to\infty} \frac{1}{N}\sum_{j=0}^{N-1}r_j$ is not valid aggregator, as the limit may not exist.)
\end{definition}

\begin{definition}[Strategies]
    A \emph{deterministic} strategy of player~$k$, for $k\in\{0,1\}$, is a function that maps every realizable \emph{history} $h_j=(s_0,a_0),(s_1,a_1),\ldots,(s_{j-1},a_{j-1}),s_j$, where $s_j\in V_k$, into an action $(s_j,p)\in A$. A \emph{randomized} strategy for player~$k$ is a probability distribution over deterministic strategies. A strategy is said to be \emph{memoryless} if it depends on the history~$h_j$ only through the current state~$s_j$. Finally, a \emph{positional} strategy is a deterministic memoryless strategy.
\end{definition}

\begin{definition}[Outcomes of strategy pairs]
    A pair of strategies $\sigma_0$ and $\sigma_1$ of players~$0$ and~$1$, respectively, also known as a \emph{strategy profile}, gives rise to a fixed distribution of plays starting from~$s$. We let $out^{\sigma_0,\sigma_1}(s)=out^{\sigma_0,\sigma_1}_\agg(s)$ be the \emph{expected} outcome of plays from~$s$ according to this distribution with respect to~$\agg$.
\end{definition}

\begin{definition}[Values of strategies]
    If $\sigma_0$ and $\sigma_1$ are strategies of players~$0$ and~$1$, respectively, the values guaranteed by strategies $\sigma_0$ and $\sigma_1$, respectively, from state~$s$ are defined as follows:
    \[ val^{\sigma_0,\cdot}(s) \EQ \inf_{\tau_1} \, out^{\sigma_0,\tau_1}(s) \;,\]
    \[ val^{\cdot,\sigma_1}(s) \EQ \sup_{\tau_0} \, out^{\tau_0,\sigma_1}(s) \;,\]
where $\tau_0$ and $\tau_1$ are arbitrary strategies of players~$0$ and~$1$, respectively.
\end{definition}

\begin{definition}[Lower and upper values]
    The \emph{lower} and \emph{upper} values of a game are defined as:
    \begin{gather*}
        \underline{val}(s) \EQ \sup_{\sigma_0} \, val^{\sigma_0,\cdot}(s) \EQ \sup_{\sigma_0} \inf_{\sigma_1} \, out^{\sigma_0,\sigma_1}(s) \;, \\
        \overline{val}(s) \EQ \inf_{\sigma_1} \, val^{\cdot,\sigma_1}(s) \EQ \inf_{\sigma_1} \sup_{\sigma_0} \, out^{\sigma_0,\sigma_1}(s) \;.
    \end{gather*}
\end{definition}
Note that $\underline{val}(s)$ and $\overline{val}(s)$ are always well defined and $\underline{val}(s)\le \overline{val}(s)$.

A non-trivial theorem \cite{Shapley53,Everett1957,Gillette1957,liggett1969stochastic,mertens1981stochastic} says that discounted and non-discounted 2-TBSGs are \emph{determined}, i.e., we always have $\underline{val}(s)= \overline{val}(s)$. The common value, denoted $val(s)$, the called the \emph{value} of the game starting from~$s$. Furthermore, both players have positional strategies that are optimal starting from all states. For more details see \cite{gamesbook}.

\subsection{Strategy Iteration}\label{sub:strategy-iteration}

Strategy Iteration (SI), or Strategy Improvement, is a natural family of algorithms that can be used to solve the games discussed in the previous section. Algorithms in the family construct a sequence of positional strategies for player~$0$, or equivalently for player~$1$. Each strategy in the sequence is obtained by applying one or more local improvement steps to the previous strategy. The nature of the games ensures that the algorithm cannot be stuck at a local optimum. As the number of positional strategies is finite, the algorithm must terminate with an optimal strategy for player~$0$.

Strategy Iteration algorithms form a generalization of the \emph{Policy Iteration} family of algorithms \cite{Howard60,Derman72,Puterman94} used to solve 1-player games such as \emph{Markov Decision Processes (MDPs)}. It may also be viewed as a generalization of the simplex algorithm used to solve linear programming (LP) problems.

Suppose that the set of states of the game is $S=S_0\cup S_1=[n]$ and that the set of states controlled by player~$0$, the maximizer, is $S_0=[n_0]$. A positional strategy for the maximizer is a tuple $\sigma=(a_1,a_2,\ldots,a_{n_0})\in \praa$. We let $val^\sigma(s)=val^{\sigma,\cdot}(s)$ be the expected outcome of the game that start from state~$s$ when the maximizer uses $\sigma$ and the minimizer uses a \emph{best response} $\tau=\tau(\sigma)$, i.e., an optimal counter strategy for the minimizer when the maximizer uses~$\sigma$. We then let $val^{\sigma}=(val^{\sigma}(1),\ldots,val^{\sigma}(n))\in\RR^n$. For every strategy~$\sigma$, $val^\sigma$ can be computed in polynomial time as this amounts to solving a 1-player game. 

For two vectors $\bu,\bv\in\RR^n$ we say that $\bu\preceq\bv$  if and only if $u_i\le v_i$, for $i\in[n]$. We say that $\bu\prec\bv$ if and only if $\bu\preceq\bv$ and $u_i<v_i$ for at least one $i\in[n]$.

\begin{definition}[Switches]
  If $\sigma=(a_1,a_2,\ldots,a_{n_0})\in \praa$ is a positional strategy of player~$0$ and $b\in A_i\sm \{a_i\}$, we let $\sigma[b]=(a_1,\ldots,a_{i-1},b,a_{i+1},\ldots,a_n)$ be the strategy obtained by \emph{switching} the action at~$i$ from $a_i$ to~$b$. We say that $b$ is an \emph{improving switch} if $val^{\sigma}\prec val^{\sigma[b]}$.
\end{definition}

It is known that for every $\sigma$ and $b$, either $val^{\sigma}\preceq val^{\sigma[b]}$ or $val^{\sigma[b]}\preceq val^{\sigma}$. For discounted 2-TBSGs, under some natural non-degeneracy assumption, either $val^{\sigma}\prec val^{\sigma[b]}$ or $val^{\sigma[b]}\prec val^{\sigma}$. Furthermore, it is known that if $\sigma$ is not an optimal strategy for the maximizer, then there is at least one action~$b$ such that $val^{\sigma}\prec val^{\sigma[b]}$. For non-discounted 2-TBSGs, things are a bit more complicated. There may not be an action~$b$ for which $val^{\sigma}\prec val^{\sigma[b]}$ even if~$\sigma$ is not an optimal strategy.

To deal with non-discounted 2-TBSGs, we can define an appropriate \emph{extended valuation} $\VAL^{\sigma}$ for which the above properties do hold. (Each component of $\VAL^{\sigma}$ is usually a pair composed of $val^{\sigma}(i)$ and a \emph{potential}, or \emph{bias}, and pairs are ordered lexicographically. See, e.g., \cite{Puterman94}.) For discounted games we can simply take $\VAL^\sigma=val^\sigma$. Also, for any strategy~$\sigma$, $\VAL^\sigma$ can be computed in polynomial time.

Even a stronger property holds. If $b_1,b_2,\ldots,b_k$ are all improving switches with respect to~$\sigma$ belonging to different states, i.e., $\VAL^{\sigma}\prec \VAL^{\sigma[b_i]}$, for $i\in[k]$, then we also have $\VAL^{\sigma}\prec \VAL^{\sigma[b_1,\ldots,b_k]}$, where $\sigma[b_1,\ldots,b_k]$ is the strategy obtained from~$\sigma$ by applying all the switches simultaneously. (Note however, that there are cases in which $\VAL^\sigma \prec \VAL^{\sigma[b_1,\ldots,b_k]}\prec \VAL^{\sigma[b_i]}$ for every $i\in [k]$.) 

This suggests a natural class of \emph{strategy iteration} algorithms for finding an optimal strategy for the maximizer. Start from any initial positional strategy~$\sigma$. If there are no improving switches with respect to~$\sigma$, then~$\sigma$ is optimal. Otherwise, find some improving switches $b_1,\ldots,b_k$ belonging to different states and move from $\sigma$ to $\sigma[b_1,\ldots,b_k]$. Note that this is a family of algorithms and not a concrete algorithms as in many cases the algorithm can choose which set of improving switches to choose.

The most natural strategy iteration algorithm is perhaps Howard's algorithm \cite{Howard60} which finds the best improving switch from each state, if any, and performs all these best improving switches simultaneously. This is a polynomial time algorithm for 2-TBSGs with a fixed discount factor \cite{Ye11a,HaMiZw13,Scherrer16}, but it may require exponential time even for undiscounted MDP, as shown by Fearnley \cite{Fearnley10} following a breakthrough exponential lower bound for Parity Games (PGs) obtained by Friedman \cite{Friedmann11}. It should be noted, however, that Howard's algorithm seems to work very well in practice, similar to the good behavior exhibited by the simplex algorithm in practice. 

The \RAR\ algorithm is a randomized strategy iteration algorithm that performs only one improving switch at a time. Intriguingly, the worst-case expected complexity of the algorithm is subexponential, far superior to all known deterministic algorithms, including those that perform many improving switches at once. 

For more on strategy iteration algorithms for games, see \cite{AkianCDG13,AugerMS21}.

\section{Pseudocode of the \RAR\ algorithm}\label{A:RAR}

Pseudocode of the \RAR\ algorithm, referred to for brevity as \RARalg, is given in Figure~\ref{F:RAR-alg}.

\begin{figure}[p]
\begin{algorithm}[H]
\DontPrintSemicolon
\SetKwProg{Fn}{Algorithm}{:}{}
\SetKwInOut{Input}{Input}
\SetKwInOut{Output}{Output}

\Fn{$\RARalg(\vec{H}\in\vec{\mathcal{H}}(\pra),\,\ba\in\pra$)}{
  \BlankLine
  \Input{An AUSO $\vec{H}$ and a start vertex $\ba\in\vec{H}$}
  \Output{The unique sink of $\vec{H}$}
  \BlankLine
  $B \gets (\bigcup_{i=1}^n A_i)\setminus \ba$\;
  \If{$B=\emptyset$}{
    \Return $\ba$\;
  }
  $r\gets \algcall{Rand}(B)$\;
  $\bb \gets \RARalg(\vec{H}\setminus r,\,\ba)$\;
  $\bc \gets \bb[r]$\;

  \If{$(\bc \to \bb)\in \vec{H}$}{
    \Return $\bb$\;
  }
  \Else{
    \Return $\RARalg(\vec{H},\,\bc)$\;
  }
}
\end{algorithm}
\caption{Pseudocode of the \RAR\ algorithm.}\label{F:RAR-alg}
\end{figure}

\begin{figure}[p]
\[
\renewcommand{\arraystretch}{1.3}
\begin{array}{c|*{9}{c}}
 & 0 & 1 & 2 & 3 & 4 & 5 & 6 & 7 & 8 \\
\hline
0 & 1 & 1 & 1 & 1 & 1 & 1 & 1 & 1 & 1 \\
1 & 2 & \frac{5}{2} & \frac{17}{6} & \frac{37}{12} & \frac{197}{60} & \frac{69}{20} & \frac{503}{140} & \frac{1041}{280} & \frac{9649}{2520} \\
2 & \frac{7}{2} & \frac{14}{3} & \frac{45}{8} & \frac{773}{120} & \frac{322}{45} & \frac{9817}{1260} & \frac{84323}{10080} & \frac{4267}{480} & \frac{472373}{50400} \\
3 & \frac{17}{3} & \frac{185}{24} & \frac{48}{5} & \frac{545}{48} & \frac{65461}{5040} & \frac{24391}{1680} & \frac{206819}{12960} & \frac{15711797}{907200} & \frac{185713321}{9979200} \\
4 & \frac{209}{24} & \frac{713}{60} & \frac{10843}{720} & \frac{45827}{2520} & \frac{24467}{1152} & \frac{1255151}{51840} & \frac{49176893}{1814400} & \frac{597058757}{19958400} & \frac{3909784957}{119750400}
\end{array}
\]
\vspace*{1cm}
\[
\renewcommand{\arraystretch}{1.1}
\begin{array}{c|*{9}{c}}
 & 0 & 1 & 2 & 3 & 4 & 5 & 6 & 7 & 8 \\
\hline
0 & 1.0000 & 1.0000 & 1.0000 & 1.0000 & 1.0000 & 1.0000 & 1.0000 & 1.0000 & 1.0000 \\
1 & 2.0000 & 2.5000 & 2.8333 & 3.0833 & 3.2833 & 3.4500 & 3.5929 & 3.7179 & 3.8290 \\
2 & 3.5000 & 4.6667 & 5.6250 & 6.4417 & 7.1556 & 7.7913 & 8.3654 & 8.8896 & 9.3725 \\
3 & 5.6667 & 7.7083 & 9.6000 & 11.3542 & 12.9883 & 14.5185 & 15.9583 & 17.3190 & 18.6100 \\
4 & 8.7083 & 11.8833 & 15.0597 & 18.1853 & 21.2387 & 24.2120 & 27.1037 & 29.9152 & 32.6495
\end{array}
\]

\caption{The values of $F(n,k)$ for $0\le n\le 4$ and $0\le k\le 8$ in exact and approximate forms.}
    \label{fig:F-values}
\end{figure}

\section{Some values of the function $F(n,k)$}\label{A:F}

Some values of the function $F(n,k)$ are given in Figure~\ref{fig:F-values} of Appendix~\ref{A:F}.

\section{Equivalent formulation of the recurrence relation}\label{A:equiv}

Although not strictly necessary for our analysis, it is instructive to show that the recurrence relation of Definition~\ref{D:F(n,k)} can be expressed very succinctly in terms of the discrete difference operators~$\Delta_n$ and~$\Delta_k$ defined as follows:
\begin{align*}
    \Delta_n F(n,k) &\EQ F(n,k) - F(n-1,k)\;, \\
    \Delta_k F(n,k) &\EQ F(n,k) - F(n,k-1)\;.
\end{align*}
Equation~\eqref{F(n,k)-1} of Definition~\ref{D:F(n,k)} is then equivalent to:
\begin{equation}
  (n+k)\,\Delta_k F(n,k) \EQ \sum_{i=1}^n F(n-i,k) \;. \label{eq:diff}
\end{equation}

By applying $\Delta_k$ to this equality we obtain:
\begin{lemma}\label{L:diff}
    The function $F(n,k)$ from Definition~\ref{D:F(n,k)} satisfies the following difference relation
    \[ (n+k)\,\Delta_n\Delta_k F(n,k) \EQ F(n-1,k-1) \;.\]
\end{lemma}

\begin{proof}
    We use the following two simple and easily verifies rules for computing finite differences which are analogous to the familiar differentiation rule. First, $\Delta_n \sum_{i=1}^n F(n-i,k) = F(n-1,k)$. Second, if $A(n)$ and $B(n)$ are two functions, then $\Delta_n A(n)B(n) = A(n)\Delta_n B(n) + (\Delta_n A(n))B(n-1)$. Applying these rules, with $A(n)=n+k$ and $B(n)=\Delta_k F(n,k)$, to
    \eqref{eq:diff}, and noting that $\Delta_n A(n)=1$, we get:
    \[ (n+k)\,\Delta_n\Delta_k F(n,k) + \Delta_k F(n-1,k) \EQ F(n-1,k) \;,\]
    which is equivalent to the claim of the lemma.
\end{proof}

Since
\begin{align*}
    \Delta_n\Delta_k F(n,k) & \EQ \Delta_k F(n,k) - \Delta_k F(n-1,k) \\
    &\EQ (F(n,k)-F(n,k-1)) - (F(n-1,k)-F(n-1,k-1) \;,
\end{align*}
we immediately get the following corollary:
\begin{corollary}\label{C-finite}
    The function $F(n,k)$ of Definition~\ref{D:F(n,k)} satisfies the following finite recurrence relation:
    \[ F(n,k) \EQ F(n,k-1) + F(n-1,k) -\left(1-\frac{1}{n+k}\right)F(n-1,k-1) \;.\]
\end{corollary}

The recurrence relation in Lemma~\ref{L:diff} is discrete analog of the interesting partial differential equation $(x+y)\frac{\partial^2 F(x,y)}{\partial x\partial y}=F(x,y)$. 

Although the recurrence relations in Lemma~\ref{L:diff} and Corollary~\ref{C-finite} look simpler, the original recurrence relation of Definition~\ref{D:F(n,k)} has the important advantage of being \emph{monotone}, i.e., it does not involve subtractions. This makes the original definition more suited for obtaining upper bounds on~$F(n,k)$.

\section{Solving the recurrence for $F(n,0)$}\label{A:Binary}

For brevity, we let $F(n)=F(n,0)$. Note that $F(n)$ is an upper bound on the length of the path constructed by the \RAR\ algorithm when called on a binary AUSOs with $n$ states. The function $F(n)$ satisfies the following recurrence relation.

\[ F(n) \;\le\; F(n-1) + \frac{1}{n}\sum_{i=1}^n F(n-i) \quad,\quad n>0 \;, \]
\[ F(0) \EQ 1\]

The \RAR\ algorithm for binary AUSOs was studied by G{\"a}rtner \cite{Gartner02}. (By analogy with randomized simplex algorithms, G{\"a}rtner \cite{Gartner02} calls it the \emph{Random-Facet} algorithm.) In particular, he obtains a closed-form expression for $F(n)$, that can be easily proved by induction, and the following upper bound:

\begin{lemma}[G{\"a}rtner \cite{Gartner02}]\label{L:Fa}
    \[ F(n) \EQ \sum_{r=0}^n \binom{n}{r}\frac1{r!} \LE \sum_{r=0}^n \frac{n^r}{(r!)^2} \LE \left(\sum_{r=0}^n \frac{n^{r/2}}{r!}\right)^2 \LE e^{2\sqrt{n}}\;.\]
\end{lemma}

The bound of Lemma~\ref{L:Fa} can be improved, by a small polynomial factor, as follows:
\[ F(n) \EQ \sum_{r=0}^n \binom{n}{r}\frac1{r!} \LE \sum_{r=0}^n \frac{n^r}{(r!)^2} \LE I_0(2\sqrt{n}) \EQ \frac{e^{2\sqrt n}}{2\sqrt\pi\,n^{1/4}}\left(1+O(n^{-1/2})\right)\;,\]
where
\[ I_0(x) \EQ \sum_{r=0}^\infty \frac{1}{(r!)^2}\left(\frac{x}{2}\right)^{2r} \EQ \sum_{r=0}^{\infty}\frac{x^{2r}}{4^r\,(r!)^2}\]
is the modified Bessel function \cite[Chapter 10]{NISTDLMF}. 

It is interesting to note that $F(n)$ is also the expected number of increasing subsequences, of all possible lengths, of a random permutation. Lifschitz and Pittel \cite{LiPi81} conducted an even more accurate study of $F(n)$ and proved, among other things, that
\[
F(n) \EQ \frac{e^{2\sqrt n}}{2\sqrt{\pi e}\,n^{1/4}}\left(1+O\!\left(n^{-1/2}\right)\right),
\qquad n\to\infty.
\]
Note that this improves on the previous bound by a constant factor of $e^{-1/2}$. The same asymptotic expression also appears in Flajolet and Sedgewick \cite{FlajoletSedgewick2009AnalyticCombinatorics}, Eq.\ (88) on page 596, where it is claimed that it can be obtained using standard generating functions techniques.

It is also interesting to note that $F(n)=L_n(-1)$, where $L_n(x)$ is the $n$-th \emph{Laguerre polynomial} \cite[Chapter 18]{NISTDLMF}.. The Laguerre polynomials are defined as follows:
\[ L_n(x) \EQ \sum_{r=0}^n \binom{n}{r}\frac{(-x)^r}{r!} \;.\]
They also satisfy the finite recurrence:
\[ L_{n+1}(x) \EQ \frac{(2n+1-x)L_n(x)-nL_{n-1}(x)}{n+1} \;,\]
\[ L_1(x) \EQ 1-x \quad,\quad L_0(x) \EQ 1 \;.\]
This of course also leads to a finite recurrence, with polynomial coefficients, for $F(n)$. This finite recurrence can be easily proved by induction without referring to Laguerre polynomials. For the asymptotics of Laguere polynomials, see, e.g., Frenzen and Wong \cite{FrWo88}.

\section{A simple upper bound on $F(n,k)$}\label{A:Simple}

In this short section we show that the upper bound $F(n)\le e^{2\sqrt{n}}$ can be used to obtain a simple, non-trivial, and sometimes asymptotically tight, upper bound on $F(n,k)$.

\begin{lemma}
  For every $n,k\ge 0$ we have:
  \[ F(n,k) \LE F(n+k) \LT e^{2\sqrt{n+k}} \;.\]
\end{lemma}

\begin{proof}
    The proof is by induction on $n$ and $k$. As base cases we have $F(n,0)=F(n)$, for every $0\le n$, and $F(0,k)=1\le F(k)$, for every $0\le k$. Assume now that $n,k>0$. Let $m=n+k$. By the recurrence relation and the induction hypothesis we have:
    \begin{align*}
    F(n,k) & \EQ F(n,k-1) + \frac1{n+k}\sum_{i=1}^n F(n-i,k) \\
    & \LE F(m-1) + \frac1m\sum_{i=1}^n F(m-i) \\
    & \LE F(m-1) + \frac1m\sum_{i=1}^m F(m-i) \\
    & \EQ F(m) \EQ F(n,k) \;. \qedhere
    \end{align*}
\end{proof}

We believe that for $0\le k\le n$, the bound $F(n,k)\le e^{2\sqrt{n+k}}$ is asymptotically tight, up to polynomial factors, but becomes loose when $k>n$.

\section{Arriving at the definitions of $U(n,k)$}\label{app:U}

In this section we sketch the heuristic that lead us, with the help of Google's Gemini, to reach the definition of $U(n,k)$ used in Section~\ref{S:Recurrence}. We expect the solution of the recurrence relation to exhibit an asymptotic behavior of the form $e^{S(n,k)}$. It is thus natural to find out which properties of $S(n,k)$ would enable us to prove by induction that $F(n,k)\le U(n,k)=e^{S(n,k)}$, as we eventually managed to do in Section~\ref{S:Recurrence}. The key inequality that needs to be satisfied for the induction to go through is inequality \eqref{eq:Q} of that section that we repeat here:
\begin{equation}
    \frac{U(n,k-1)}{U(n,k)} + \frac{1}{(n+k)\,U(n,k)}\sum_{i=1}^{n}U(n-i,k) \LE 1 \quad,\quad 0<n<k \;. \label{eq:Q2}
\end{equation}

We now use the following rough estimates, which are converted in Section~\ref{S:Recurrence} into rigorous bounds. We expect the estimates below to be fairly accurate as $S(n,k)$ is expected to be slowly increasing in~$n$ and~$k$, and hence the partial derivatives $S_n$ and $S_k$ are expected to be small.
    \[ \frac{U(n,k-1)}{U(n,k)} \EQ e^{S(n,k-1) - S(n,k)} \;\approx\; e^{-S_k} \;\approx\; 1-S_k \;.\]
    \[ \frac{U(n-i, k)}{U(n,k)} \approx  e^{-i S_n} \]
\[ \sum_{i=1}^n \frac{U(n-i,k)}{U(n,k)} \approx \sum_{i=1}^n e^{-i S_n} \;\approx\; \frac{e^{-S_n}}{1 - e^{-S_n}} \EQ \frac{1}{e^{S_n} - 1} \;\approx\; \frac{1}{S_n} \;.\]

Thus to satisfy \eqref{eq:Q2} is a tight way, we want
\[ (1-S_k) + \frac{1}{(n+k)S_n} \;\approx\; 1 \;.\]
Simplifying and turning this into an equality we get the partial differential equation
\begin{equation}
    (n+k)S_nS_k \EQ 1 \;, \label{eq:PDE}
\end{equation}
which is exactly the partial differential equation of Lemma~\ref{L:PDE}. This is the key equation that drives the whole analysis.

Based on numerical experiments it was reasonable to stipulate that $S(n,k)=2\sqrt{n\Phi(\frac{k}{n})}$. By differentiating, we get
\[ S_n \EQ \frac{\Phi - \alpha \Phi'}{\sqrt{n \Phi}} \quad,\quad S_k \EQ \frac{\Phi'}{\sqrt{n \Phi}} \;,\]
exactly as in Lemma~\ref{L:SnSk}. Plugging these expressions into \eqref{eq:PDE} we get
\[n(1+\alpha)  \frac{\Phi - \alpha \Phi'}{\sqrt{n\Phi}}  \frac{\Phi'}{\sqrt{n\Phi}}  \EQ 1 \;,\]
which is immediately simplified into the differential equation \eqref{eq:Phi} of Section~\ref{S:Recurrence}.

\section{Properties of $\Phi(\alpha)$}\label{app:Phi}

\subsection{Existence and uniqueness of $\Phi(\alpha)$}

Our first goal is to prove that $\Phi(\alpha)=\Phi_-(\alpha)$, the solution of \eqref{eq:Phi2} with the minus sign, i.e., 
\begin{equation}
    \Phi' \EQ f(\alpha,\Phi) \EQ \frac{\Phi - \sqrt{\Phi^2 - \frac{4\alpha}{1+\alpha}\Phi}}{2\alpha} \quad,\quad \Phi(1)=2\;, \label{eq:Phi3}
\end{equation}
exists and is unique for every $\alpha\ge 1$.

In principle, this is an \emph{Initial Value Problem} (IVP) of the form $x'=f(t,x)$, $x(t_0)=x_0$. If $f$ is Lipschitz, the existence and uniqueness of the solution $x$ follows immediately from a standard theorem of Picard and Lindel\"{o}f \cite[Theorem~2.2]{Teschl12}. We are faced, however, with two technical difficulties. First, due to the square root appearing in our $f(\alpha,\Phi)$, the function is not Lipschitz at $(1,2)$. (The partial derivative $\frac{\partial f}{\partial\Phi}\to -\infty$ as $\Phi\to 2^+$). Second, $f$ is only well-defined when $\Phi\ge \frac{4\alpha}{1+\alpha}$ so we need to show that the solution whose existence we are trying to prove satisfies this condition.

To overcome the second difficulty, we define a modified function 
\[ \bar{f}(\alpha,\Phi) \EQ \frac{\Phi - \sqrt{\left(\Phi^2 - \frac{4\alpha}{1+\alpha}\Phi\right)^+}}{2\alpha} \;, \] 
where $x^+=\max\{x,0\}$. Now $\bar{f}(\alpha,\Phi)$ is well-defined and continuous for every $\alpha\ge 1$ and every~$\Phi$.
The existence of at least one solution defined for every $\alpha\ge 1$ then follows from Peano existence theorem \cite[Theorem~2.19]{Teschl12}, which only requires the continuity of~$\bar{f}$, and a standard extension criterion \cite[Corollary~2.15]{Teschl12}. To prove uniqueness, we temporarily consider $\bar{f}^+$ instead of $\bar{f}$, as $\bar{f}^+$ satisfies a one-sided Lipschitz condition $\bar{f}^+(t,\phi_1)-\bar{f}^+(t,\phi_2)\le L(\phi_1-\phi_2)$, for $\phi_1>\phi_2$, where $L=\frac{1}{2}$. It then follows that the initial value problem $\Phi'=\bar{f}^+(\alpha,\Phi)$, $\Phi(1)=2$ has a unique solution defined for every $\alpha\ge 1$.  (See \cite[Problem~2.25]{Teschl12} and \cite[\S 9 X]{Walter98}.)

To show that the unique solution of $\Phi'=\bar{f}^+(\alpha,\Phi)$, $\Phi(1)=2$ is also the unique solution of $\Phi'=f(\alpha,\Phi)$, $\Phi(1)=2$, we need to show that $\Phi(\alpha)\ge g(\alpha)=\frac{4\alpha}{1+\alpha}$. Relying on the comparison principle \cite[Lemma 1.2]{Teschl12}, it is enough to check that $g'(\alpha)\le \bar{f}^+(\alpha,g(\alpha))=f(\alpha,g(\alpha))$. This holds as 
\[ g'(\alpha) \EQ \left(\frac{2}{1+\alpha}\right)^2 \quad,\quad f(\alpha,g(\alpha))\EQ \frac{g(\alpha)}{2\alpha}\EQ \frac{2}{1+\alpha} \;,\]
and $\frac{2}{1+\alpha}\le 1$, for $\alpha\ge 1$.

\subsection{Upper bound on $\Phi(\alpha)$}

Let
\[ B(\alpha) \EQ 2 + \s + \ln(\s+1) \EQ 2 + \ln\frac{1+\alpha}{2} + \ln\left(\ln\frac{1+\alpha}{2}+1\right) \quad,\quad \s\EQ \ln\frac{1+\alpha}{2}\;.\]
Our goal is to prove that $\Phi(\alpha)\le B(\alpha)$, for every $\alpha\ge 1$. Using the comparison principle, as above, it is enough to show that
\[ B' \GE f(\alpha,B) \EQ \frac{B-\sqrt{B^2-g(\alpha)B}}{2\alpha} \quad,\quad \alpha\ge 1\;,\]
or equivalently that 
\[ 2\alpha B' - B \GE -\sqrt{B^2-g(\alpha)B} \quad,\quad \alpha\ge 1\;.\]
It is thus enough to show that
\[ (2\alpha B' - B)^2 \LE B^2-g(\alpha)B \quad,\quad \alpha\ge 1\;, \]
which simplifies to
\[ P \EQ \alpha (B')^2 -BB' + \frac{B}{1+\alpha} \LE 0 \quad,\quad \alpha\ge 1\;. \]
Note that this is exactly the form of Equation \eqref{eq:Phi} of Section~\ref{S:Recurrence}. A simple calculation shows that 
\[ B'(\alpha) \EQ \frac{1}{1+\alpha} + \frac{1}{(1+\alpha)(\s+1)} \EQ \frac{1}{1+\alpha}\, \frac{\s+2}{\s+1}\;.\]
Plugging this expression for $B'$ into the definition of~$P$, we get
\begin{gather*}P  \EQ \alpha \left( \frac{1}{1+\alpha} \left( \frac{\s+2}{\s+1} \right) \right)^2 - B \left( \frac{1}{1+\alpha} \left( \frac{\s+2}{\s+1} \right) \right) + \frac{B}{1+\alpha}\\
\EQ
\frac{1}{1+\alpha} \left( \frac{\alpha}{1+\alpha} \frac{(\s+2)^2}{(\s+1)^2} - \frac{B}{\s+1} \right) 
\EQ \frac{1}{(1+\alpha)(\s+1)} \left( \frac{\alpha}{1+\alpha} \frac{(\s+2)^2}{\s+1} - B \;.\right)
\end{gather*}
As $\s=\ln\frac{1+\alpha}{2}$, we have $\alpha=2e^\s-1$ and hence
$\frac{\alpha}{1+\alpha} \EQ \frac{2e^\s - 1}{2e^\s} \EQ 1 - \frac{1}{2}e^{-\s}$.
Thus,
\[ P \EQ \frac{1}{(1+\alpha)(\s+1)} \left( \left(1 - \frac{1}{2}e^{-\s}\right) \frac{(\s+2)^2}{\s+1} - B \right) \;.\]
To prove that $P\le 0$, we need to show that
\[ W(\s) = B(\s) - \left(1 - \frac{1}{2}e^{-\s}\right) \frac{(\s+2)^2}{\s+1} \GE 0 \;.\]
At $\s=0$, we have $W(0)=0$. It is thus enough to show that $W'(\s)\ge 0$, for $\s\ge 0$. A tedious calculation shows that
\[ W'(\s) = \frac{\s+2}{2e^\s(\s+1)^2} \Big[ 2e^\s - (\s^2 + 2\s + 2) \Big] \GE 0\;,\]
where the inequality follows as $e^\s\ge 1+\s+\frac{\s^2}{2}$, for every $\s\ge 0$. This completes the proof of the upper bound.

\subsection{Lower bound on $\Phi(\alpha)$}
Let
\[ L(\alpha) \EQ 2+\s+\ln(1+\s) - \left(\frac{\s}{\s+1}\right)^3 \EQ 2 + \ln\frac{1+\alpha}{2} + \ln\left(1+\ln\frac{1+\alpha}{2}\right) - \left(\frac{\ln\frac{1+\alpha}{2}}{1+\ln\frac{1+\alpha}{2}}\right)^3\;.\]
An even more tedious calculation shows that 
\[ L' \LE f(\alpha,L) \quad,\quad \alpha\ge 1\;,\]
which implies that $L(\alpha)\le \Phi(\alpha)$, for every $\alpha\ge 1$.

Note that the upper and lower bounds $B(\alpha)$ and $L(\alpha)$ differ by at most~$1$. There particular forms of $B(\alpha)$ and $L(\alpha)$ were chosen because they additionally satisfy $L(1)=\Phi(1)=B(1)=2$ and $L'(1)=\Phi'(1)=B'(1)=1$.

\subsection{Plot of $\Phi(\alpha)$ and its lower and upper bounds}

\begin{figure}
    \centering
    \includegraphics[width=0.9\linewidth]{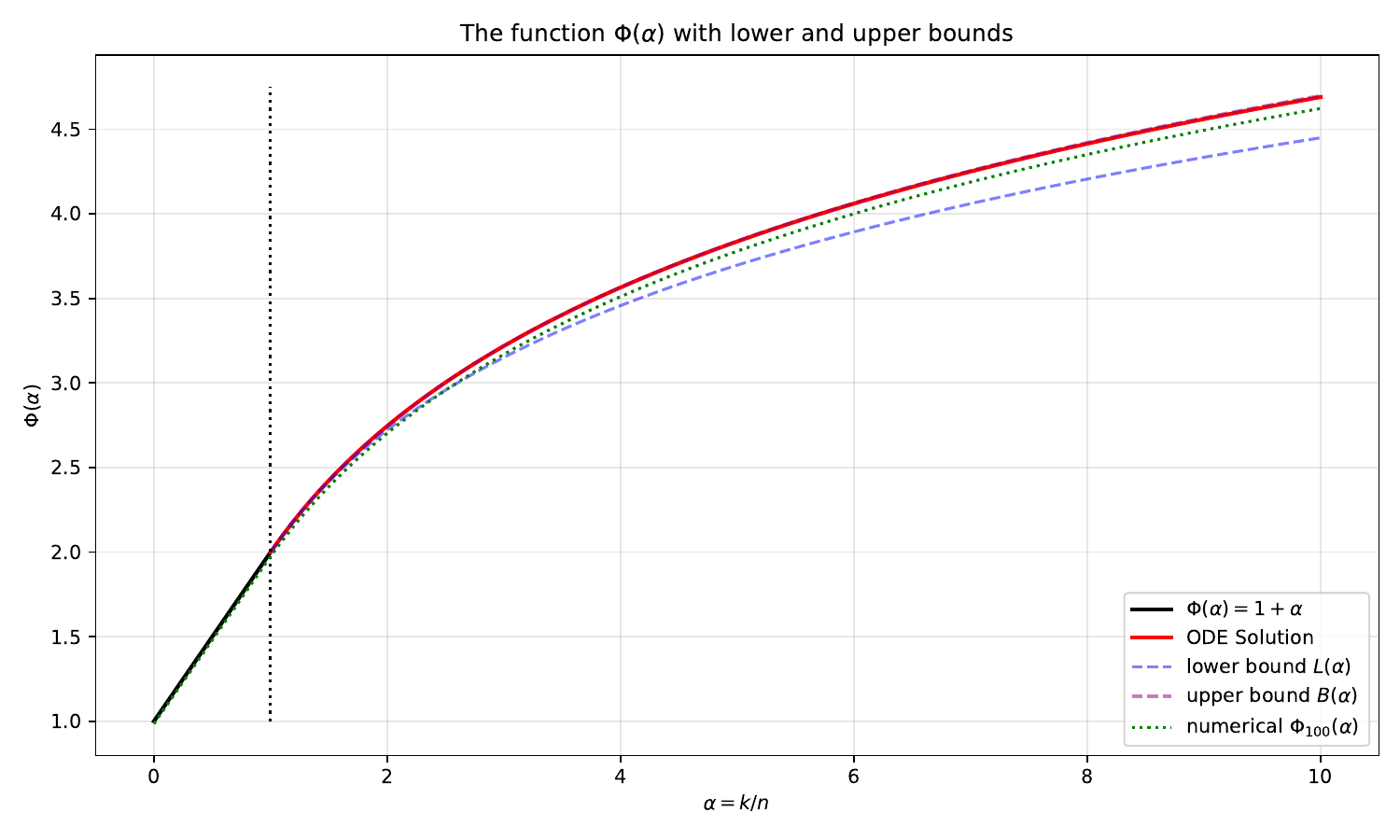}
    \caption{A plot of $\Phi(\alpha)$ and related functions.}
    \label{F:plot-Phi}
\end{figure}


A plot depicting the functions $\Phi(\alpha)$ and its lower and upper bounds $L(\alpha)$ and $B(\alpha)$ is given in Figure~\ref{F:plot-Phi}. The values of $\Phi(\alpha)$ were obtained by numerically solving the differential equation defining~$\Phi(\alpha)$. As claimed, $L(\alpha)$ is a lower bound, and~$B(\alpha)$ is an upper bound on~$\Phi(\alpha)$. The upper bound $B(\alpha)$ seems to be extremely tight, at least for $\alpha\le 10$. Also shown in the plot, as a dotted green line, is the empirical function $\tilde{\Phi}_{100}(\alpha) = \left(\frac{\ln F(n,\alpha n)}{2\sqrt{n}}\right)^2$, for $n=100,000$. As expected, this function lies below $\Phi(\alpha)$, consistent with the upper bound claim. The small gap between $\tilde{\Phi}_{100}(\alpha)$ and $\Phi(\alpha)$ is explained by the hypothesis that $F(n,\alpha n)$ is asymptotically of the form $(n(1+\alpha))^{A}\,e^{2\sqrt{n\Phi(\alpha)}}$ for some constant~$A<0$. Choosing $A=-\frac{1}{2}$ yields an almost perfect fit. The choice $A=-\frac{1}{2}$ is also supported by some heuristic calculations that we do not repeat here. The factor $n+k$ used in our upper bound proof, which corresponds to the choice $A=1$, was added to simplify the proof. Proving that the $n+k$ factor can possibly be replaced by a $(n+k)^{-1/2}$ would require a more complicated analysis that would use more terms in the Taylor expansions used in the current proof.

\subsection{Concavity of $\Phi$}

To prove the strict concavity of $\Phi(\alpha)$ for $\alpha>1$, we prove that $\Phi''(\alpha)<0$, for $\alpha>1$. Recall that $\Phi$ satisfies Equations~\eqref{eq:Phi} and~\eqref{eq:Phi2} from Section~\ref{S:Recurrence}:
\begin{equation*}
    \alpha (\Phi')^2 - \Phi \Phi' + \frac{\Phi}{1+\alpha} \EQ 0 \quad,\quad
    \Phi' \EQ \frac{\Phi - D}{2\alpha} \quad,\quad D \EQ \sqrt{\Phi^2 - \frac{4\alpha\Phi}{1+\alpha}} \;.
\end{equation*}
Differentiating the first equation with respect to $\alpha$ yields:
\begin{equation*}
    (\Phi')^2 + 2\alpha\Phi'\Phi'' - \left(\Phi'\Phi' + \Phi\Phi''\right) + \frac{\Phi'(1+\alpha) - \Phi}{(1+\alpha)^2} \EQ 0 \;.
\end{equation*}
After simplifying we get:
\begin{equation*}
    \Phi''(2\alpha\Phi' - \Phi) \EQ \frac{\Phi - (1+\alpha)\Phi'}{(1+\alpha)^2} \;.
\end{equation*}
Note that $2\alpha\Phi' - \Phi=-D<0$, for $\alpha>1$. To finish the proof, note that
\[ \Phi - (1+\alpha)\Phi' \EQ \Phi - (1+\alpha) \frac{\Phi - D}{2\alpha} \EQ \frac{\alpha-1}{2\alpha}\Phi + \frac{1+\alpha}{2\alpha}D \GT 0\;,\]
as both $\Phi>0$ and $D>0$, for $\alpha>1$.

\section{Partial derivatives of $S(n,k)$}\label{A:partial}

We give the explicit computation of $S_n$ and $S_k$. The computations of $S_{nn}$ and $S_{kk}$ are similar.  Recall that $S(n,k)=2\sqrt{n\Phi(\alpha)}=2\sqrt{n\Phi(\frac{k}{n})}$. As before, we use the shorthands $\Phi=\Phi(\alpha)$ and $\Phi'=\Phi'(\alpha)$.

    Using the chain rule and $\frac{\partial \alpha}{\partial k} = \frac{1}{n}$, we have
    \[ S_k \EQ \frac{\partial}{\partial k} \left( 2\sqrt{n} \sqrt{\Phi} \right) \EQ 2\sqrt{n} \cdot \frac{1}{2\sqrt{\Phi}} \cdot \Phi' \cdot \frac{1}{n} \EQ \frac{\Phi'}{\sqrt{n \Phi}} \;. \]
    Using $\frac{\partial \alpha}{\partial n} = -\frac{k}{n^2} = -\frac{\alpha}{n}$, we have
\begin{align*}
    S_n &\EQ \frac{\partial}{\partial n} \left( 2\sqrt{n\Phi} \right) = 2 \left( \frac{\sqrt{\Phi}}{2\sqrt{n}} + \sqrt{n} \frac{\Phi'}{2\sqrt{\Phi}} \cdot\left( -\frac{\alpha}{n} \right) \right) \\
    &\EQ \frac{\Phi}{\sqrt{n\Phi}} - \frac{\alpha \Phi'}{\sqrt{n\Phi}} \EQ \frac{\Phi - \alpha \Phi'}{\sqrt{n\Phi}} \;. \qedhere
    \end{align*}

Since $S(n,k)$ is positively homogeneous of degree $1/2$, i.e., $S(tn,tk)=t^{1/2}S(n,k)$, it follows from Euler's Theorem on Homogeneous Functions (see, e.g., Apostol \cite{Apostol1969calculus}) that $n S_n + k S_k = \frac{1}{2}S = \sqrt{n\Phi}$, as we can also verify directly.

\section{The concavity of $S(n,k)$}\label{A:concave}

Recall that
\[
          S_{nn} \EQ \frac{\alpha^2 \Phi \Phi'' - \frac{1}{2}(\Phi - \alpha\Phi')^2}{(n\Phi)^{3/2}} \quad,\quad
    S_{kk} \EQ \frac{\Phi \Phi'' - \frac{1}{2}(\Phi')^2}{(n\Phi)^{3/2}}  \;.
\]
For $\alpha>1$ we have $\Phi>0$ and $\Phi''<0$. We also have $(\Phi - \alpha\Phi')^2\ge 0$ and $(\Phi')^2\ge 0$, Thus $S_{nn}<0$ and $S_{kk}<0$, as required.

\newpage
\restoregeometry          
\pagestyle{plain}         

\printbibliography

@string{infcomp={Information and Computation}}

@string{jacm={Journal of the ACM}}

@string{rsa={Random Structures and Algorithms}}

@string{sicomp={SIAM Journal on Computing}}

@string{tcs={Theoretical Computer Science}}

@string{ucompmath={USSR Computational Mathematics and Mathematical Physics}}

@string{focs91={Proceedings of the 32nd FOCS}}

@string{focs01={Proc.\ of 42th FOCS}}

@string{icalp10={Proc.\ of 37th ICALP}}

@string{isaac09={Proc.\ of 20th ISAAC}}

@string{lics17={Proc.\ of 32nd LICS}}

@string{soda06={Proc.\ of 17th SODA}}

@string{stoc92={Proc.\ of 24th STOC}}

@inproceedings{AnMi09,
  author    = {Daniel Andersson and Peter Bro Miltersen},
  title     = {The Complexity of Solving Stochastic Games on Graphs},
  booktitle = isaac09,
  year      = {2009},
  pages     = {112--121},
  doi       = {10.1007/978-3-642-10631-6_13}
}

@article{BjVo05,
  author    = {Henrik Bj{\"o}rklund and Sergei Vorobyov},
  title     = {Combinatorial structure and randomized subexponential algorithms for infinite games},
  journal   = tcs,
  volume    = {349},
  number    = {3},
  year      = {2005},
  pages     = {347--360},
  doi       = {10.1016/j.tcs.2005.07.041}
}

@article{BjVo07,
  author       = {Henrik Bj{\"{o}}rklund and
                  Sergei G. Vorobyov},
  title        = {A combinatorial strongly subexponential strategy improvement algorithm
                  for mean payoff games},
  journal      = {Discret. Appl. Math.},
  volume       = {155},
  number       = {2},
  pages        = {210--229},
  year         = {2007},
  url          = {https://doi.org/10.1016/j.dam.2006.04.029},
  doi          = {10.1016/J.DAM.2006.04.029},
  bibsource    = {dblp computer science bibliography, https://dblp.org}
}

@article{Condon92,
  author    = {Anne Condon},
  title     = {The Complexity of Stochastic Games},
  journal   = infcomp,
  volume    = {96},
  year      = {1992},
  pages     = {203--224},
  doi       = {10.1016/0890-5401(92)90048-K}
}

@book{Derman72,
  author    = {Cyrus Derman},
  title     = {Finite state {Markov} decision processes},
  publisher = {Academic Press},
  address   = {New York},
  year      = {1972}
}

@article{EhMy79,
  author    = {Andrzej Ehrenfeucht and Jan Mycielski},
  title     = {Positional Strategies for Mean Payoff Games},
  journal   = {International Journal of Game Theory},
  volume    = {8},
  year      = {1979},
  pages     = {109--113},
  doi       = {10.1007/BF01768705}
}

@inproceedings{EmJu91,
  author    = {E. Allen Emerson and Charanjit S. Jutla},
  title     = {Tree automata, $\mu$-calculus and determinacy},
  booktitle = focs91,
  publisher = {IEEE Computer Society Press},
  year      = {1991},
  pages     = {368--377},
  doi       = {10.1109/SFCS.1991.185392}
}

@inproceedings{Fearnley10,
  author    = {John Fearnley},
  title     = {Exponential Lower Bounds for Policy Iteration},
  booktitle = icalp10,
  year      = {2010},
  pages     = {551--562},
  doi       = {10.1007/978-3-642-14162-1_46}
}

@article{Friedmann11,
  author    = {Oliver Friedmann},
  title     = {An Exponential Lower Bound for the Latest Deterministic Strategy Iteration Algorithms},
  journal   = {Logical Methods in Computer Science},
  volume    = {7},
  number    = {3},
  year      = {2011},
  doi       = {10.2168/LMCS-7(3:23)2011}
}

@phdthesis{Gartner95PhD,
  author = {Bernd G{\"a}rtner},
  title  = {Randomized Optimization by Simplex Type Methods},
  school = {Freie Universit{\"a}t Berlin},
  year   = {1995}
}

@article{Gartner95,
  author    = {Bernd G{\"{a}}rtner},
  title     = {A Subexponential Algorithm for Abstract Optimization Problems},
  journal   = sicomp,
  volume    = {24},
  number    = {5},
  year      = {1995},
  pages     = {1018--1035},
  doi       = {10.1137/S0097539793250287}
}

@article{Gartner02,
  author    = {Bernd G{\"a}rtner},
  title     = {The Random-Facet simplex algorithm on combinatorial cubes},
  journal   = rsa,
  volume    = {20},
  number    = {3},
  year      = {2002},
  pages     = {353--381},
  doi       = {10.1002/rsa.10034}
}

@article{Goldwasser95,
  author    = {Michael Goldwasser},
  title     = {A survey of linear programming in randomized subexponential time},
  journal   = {SIGACT News},
  volume    = {26},
  number    = {2},
  year      = {1995},
  pages     = {96--104},
  doi       = {10.1145/202840.202847}
}

@book{GrThWi02,
  editor    = {Erich Gr{\"a}del and Wolfgang Thomas and Thomas Wilke},
  title     = {Automata, Logics, and Infinite Games. A Guide to Current Research},
  series    = {LNCS},
  volume    = {2500},
  publisher = {Springer},
  year      = {2002},
  doi       = {10.1007/3-540-36387-4}
}

@article{GuKaKh88,
  author  = {Vladimir A. Gurvich and Alexander V. Karzanov and Leonid G. Khachiyan},
  title   = {Cyclic Games and an Algorithm to Find Minimax Cycle Means in Directed Graphs},
  journal = ucompmath,
  volume  = {28},
  year    = {1988},
  pages   = {85--91},
  note    = {Translated from Zhurnal Vychislitel'noi Matematiki i Matematicheskoi Fiziki 28 (1988) 1407--1417.}
}

@article{Halman07,
  author       = {Nir Halman},
  title        = {Simple Stochastic Games, Parity Games, Mean Payoff Games and Discounted
                  Payoff Games Are All LP-Type Problems},
  journal      = {Algorithmica},
  volume       = {49},
  number       = {1},
  pages        = {37--50},
  year         = {2007},
  url          = {https://doi.org/10.1007/s00453-007-0175-3},
  doi          = {10.1007/S00453-007-0175-3},
  bibsource    = {dblp computer science bibliography, https://dblp.org}
}

@article{HaMiZw13,
  author    = {Thomas Dueholm Hansen and Peter Bro Miltersen and Uri Zwick},
  title     = {Strategy Iteration Is Strongly Polynomial for 2-Player Turn-Based Stochastic Games with a Constant Discount Factor},
  journal   = jacm,
  volume    = {60},
  number    = {1},
  year      = {2013},
  pages     = {1},
  doi       = {10.1145/2432622.2432623}
}

@inproceedings{HaZw15,
  author    = {Thomas Dueholm Hansen and Uri Zwick},
  title     = {An Improved Version of the Random-Facet Pivoting Rule for the Simplex Algorithm},
  booktitle = {Proceedings of the Forty-Seventh Annual {ACM} Symposium on Theory of Computing, {STOC} 2015},
  publisher = {{ACM}},
  year      = {2015},
  pages     = {209--218},
  doi       = {10.1145/2746539.2746557}
}

@book{Howard60,
  author    = {Ronald A. Howard},
  title     = {Dynamic programming and {Markov} processes},
  publisher = {MIT Press},
  year      = {1960}
}

@inproceedings{Kalai92,
  author    = {Gil Kalai},
  title     = {A Subexponential Randomized Simplex Algorithm (Extended Abstract)},
  booktitle = stoc92,
  year      = {1992},
  pages     = {475--482},
  doi       = {10.1145/129712.129759}
}

@article{Kalai92b,
  author  = {Gil Kalai},
  title   = {Upper Bounds for the Diameter and Height of Graphs of Convex Polyhedra},
  journal = {Discrete {\&} Computational Geometry},
  volume  = {8},
  year    = {1992},
  pages   = {363--372},
  doi     = {10.1007/BF02293053}
}

@article{KaKl92,
  author  = {Gil Kalai and Daniel J. Kleitman},
  title   = {Quasi-polynomial bounds for the diameter of graphs and polyhedra},
  journal = {Bull.\ Amer.\ Math.\ Soc.},
  volume  = {26},
  year    = {1992},
  pages   = {315--316},
  doi     = {10.1090/S0273-0979-1992-00285-9}
}

@article{LiPi81,
  author  = {Vladimir Lifschitz and Boris Pittel},
  title   = {The number of increasing subsequences of the random permutation},
  journal = {Journal of Combinatorial Theory, Series A},
  volume  = {31},
  number  = {1},
  year    = {1981},
  pages   = {1--20},
  doi     = {10.1016/0097-3165(81)90049-2}
}

@article{Ludwig95,
  author       = {Walter Ludwig},
  title        = {A Subexponential Randomized Algorithm for the Simple Stochastic Game
                  Problem},
  journal      = {Inf. Comput.},
  volume       = {117},
  number       = {1},
  pages        = {151--155},
  year         = {1995},
  url          = {https://doi.org/10.1006/inco.1995.1035},
  doi          = {10.1006/INCO.1995.1035},
  bibsource    = {dblp computer science bibliography, https://dblp.org}
}

@article{Matousek94,
  author       = {Ji{\v{r}}{\'i} Matousek},
  title        = {Lower Bounds for a Subexponential Optimization Algorithm},
  journal      = {Random Struct. Algorithms},
  volume       = {5},
  number       = {4},
  pages        = {591--608},
  year         = {1994},
  url          = {https://doi.org/10.1002/rsa.3240050408},
  doi          = {10.1002/RSA.3240050408},
  bibsource    = {dblp computer science bibliography, https://dblp.org}
}

@article{MaShWe96,
  author       = {Ji{\v{r}}{\'i} Matousek and
                  Micha Sharir and
                  Emo Welzl},
  title        = {A Subexponential Bound for Linear Programming},
  journal      = {Algorithmica},
  volume       = {16},
  number       = {4/5},
  pages        = {498--516},
  year         = {1996},
  url          = {https://doi.org/10.1007/BF01940877},
  doi          = {10.1007/BF01940877},
  bibsource    = {dblp computer science bibliography, https://dblp.org}
}

@book{Puterman94,
  author    = {Martin L. Puterman},
  title     = {Markov decision processes},
  publisher = {Wiley},
  year      = {1994},
  doi       = {10.1002/9780470316887}
}

@article{Shapley53,
  author  = {Lloyd S. Shapley},
  title   = {Stochastic Games},
  journal = {Proceedings of the National Academy of Sciences U.S.A.},
  year    = {1953},
  volume  = {39},
  pages   = {1095--1100},
  doi     = {10.1073/pnas.39.10.1095}
}

@inproceedings{SzWe01,
  author    = {Tibor Szab{\'o} and Emo Welzl},
  title     = {Unique Sink Orientations of Cubes},
  booktitle = focs01,
  year      = {2001},
  pages     = {547--555},
  doi       = {10.1109/SFCS.2001.959931}
}

@article{WiHo88,
  author  = {Kathy {Williamson Hoke}},
  title   = {Completely unimodal numberings of a simple polytope},
  journal = {Discrete Applied Mathematics},
  volume  = {20},
  number  = {1},
  year    = {1988},
  pages   = {69--81},
  doi     = {10.1016/0166-218X(88)90042-X}
}

@article{ZwPa96,
  author  = {Uri Zwick and Michael S. Paterson},
  title   = {The complexity of mean payoff games on graphs},
  journal = tcs,
  volume  = {158},
  pages   = {343--359},
  number  = {1--2},
  year    = {1996},
  doi     = {10.1016/0304-3975(95)00188-3}
}

@book{FlajoletSedgewick2009AnalyticCombinatorics,
  author    = {Philippe Flajolet and Robert Sedgewick},
  title     = {Analytic Combinatorics},
  publisher = {Cambridge University Press},
  address   = {Cambridge},
  year      = {2009},
  doi       = {10.1017/CBO9780511801655}
}

@software{NISTDLMF,
  author       = {Frank W. J. Olver and Adri B. Olde Daalhuis and Daniel W. Lozier and Barry I. Schneider and Ronald F. Boisvert and Charles W. Clark and Bruce R. Miller and Bonita V. Saunders and Howard S. Cohl and Mary A. McClain},
  title        = {{NIST Digital Library of Mathematical Functions}},
  howpublished = {\url{https://dlmf.nist.gov/}},
  year         = {2024},
  version      = {1.2.0}
}

@article{GaMoRu08,
  author  = {Bernd G{\"a}rtner and Walter D. Morris Jr. and Leo R{\"u}st},
  title   = {Unique Sink Orientations of Grids},
  journal = {Algorithmica},
  volume  = {51},
  number  = {2},
  year    = {2008},
  pages   = {200--235},
  doi     = {10.1007/s00453-007-9090-x}
}

@inproceedings{GaSc06,
  author    = {Bernd G{\"a}rtner and Ingo Schurr},
  title     = {Linear programming and unique sink orientations},
  booktitle = soda06,
  year      = {2006},
  pages     = {749--757},
  doi       = {10.1145/1109557.1109639}
}

@book{gamesbook,
  author    = {Nathana{\"e}l Fijalkow and Nathalie Bertrand and Patricia Bouyer-Decitre and Romain Brenguier and Arnaud Carayol and John Fearnley and Hugo Gimbert and Florian Horn and Rasmus Ibsen-Jensen and Nicolas Markey and Benjamin Monmege and Petr Novotn{\'y} and Mickael Randour and Ocan Sankur and Sylvain Schmitz and Olivier Serre and Mateusz Skomra},
  editor    = {Nathana{\"e}l Fijalkow},
  title     = {Games on Graphs},
  publisher = {Online},
  year      = {2025}
}

@book{BenderOrszag1999,
  author    = {Carl M. Bender and Steven A. Orszag},
  title     = {Advanced Mathematical Methods for Scientists and Engineers I: Asymptotic Methods and Perturbation Theory},
  publisher = {Springer},
  address   = {New York},
  year      = {1999},
  doi       = {10.1007/978-1-4757-3069-2}
}

@article{FrWo88,
  author  = {Colin L. Frenzen and Roderick Wong},
  title   = {Uniform Asymptotic Expansions of {Laguerre} Polynomials},
  journal = {SIAM Journal on Mathematical Analysis},
  volume  = {19},
  number  = {5},
  pages   = {1232--1248},
  year    = {1988},
  doi     = {10.1137/0519087}
}

@book{Apostol1969calculus,
  author    = {Tom M. Apostol},
  title     = {Calculus, Volume {II}: Multi-Variable Calculus and Linear Algebra with Applications},
  edition   = {2nd},
  year      = {1969},
  publisher = {John Wiley \& Sons},
  address   = {New York}
}

@article{AkGaGu12,
  author  = {Marianne Akian and St{\'{e}}phane Gaubert and Alexander E. Guterman},
  title   = {Tropical Polyhedra are Equivalent to mean Payoff Games},
  journal = {Int. J. Algebra Comput.},
  volume  = {22},
  number  = {1},
  year    = {2012},
  doi     = {10.1142/S0218196711006674}
}

@article{ABGJ14,
  author  = {Xavier Allamigeon and Pascal Benchimol and St{\'{e}}phane Gaubert and Michael Joswig},
  title   = {Combinatorial Simplex Algorithms Can Solve Mean Payoff Games},
  journal = {{SIAM} J. Optim.},
  volume  = {24},
  number  = {4},
  pages   = {2096--2117},
  year    = {2014},
  doi     = {10.1137/140953800}
}

@article{AGS18,
  author  = {Xavier Allamigeon and St{\'{e}}phane Gaubert and Mateusz Skomra},
  title   = {Solving generic nonarchimedean semidefinite programs using stochastic game algorithms},
  journal = {J. Symb. Comput.},
  volume  = {85},
  pages   = {25--54},
  year    = {2018},
  doi     = {10.1016/j.jsc.2017.07.002}
}

@book{ApGr11,
  author    = {Krzysztof R. Apt and Erich Gr{\"{a}}del},
  title     = {Lectures in game theory for computer scientists},
  year      = {2011},
  publisher = {Cambridge University Press}
}

@article{Ye05,
  author  = {Yinyu Ye},
  title   = {A New Complexity Result on Solving the {Markov Decision Problem}},
  journal = {Math. Oper. Res.},
  volume  = {30},
  number  = {3},
  pages   = {733--749},
  year    = {2005},
  doi     = {10.1287/moor.1050.0149}
}

@article{Ye11a,
  author  = {Yinyu Ye},
  title   = {The Simplex and Policy-Iteration Methods Are Strongly Polynomial for the {Markov Decision Problem} with a Fixed Discount Rate},
  journal = {Math. Oper. Res.},
  volume  = {36},
  number  = {4},
  pages   = {593--603},
  year    = {2011},
  doi     = {10.1287/moor.1110.0516}
}

@article{Scherrer16,
  author  = {Bruno Scherrer},
  title   = {Improved and Generalized Upper Bounds on the Complexity of Policy Iteration},
  journal = {Math. Oper. Res.},
  volume  = {41},
  number  = {3},
  pages   = {758--774},
  year    = {2016},
  doi     = {10.1287/moor.2015.0753}
}

@book{FiVr12,
  author    = {Jerzy Filar and Koos Vrieze},
  title     = {Competitive Markov decision processes},
  year      = {2012},
  publisher = {Springer Science \& Business Media}
}

@book{solan2022course,
  author    = {Eilon Solan},
  title     = {A course in stochastic game theory},
  volume    = {103},
  year      = {2022},
  publisher = {Cambridge University Press},
  doi = {10.1017/9781009029704}
}

@book{baier2008principles,
  author    = {Christel Baier and Joost-Pieter Katoen},
  title     = {Principles of model checking},
  year      = {2008},
  publisher = {MIT press}
}

@article{CaludeJKLS22,
  author  = {Cristian S. Calude and Sanjay Jain and Bakhadyr Khoussainov and Wei Li and Frank Stephan},
  title   = {Deciding Parity Games in Quasi-polynomial Time},
  journal = {{SIAM} J. Comput.},
  volume  = {51},
  number  = {2},
  pages   = {17--152},
  year    = {2022},
  doi     = {10.1137/17M1145288}
}

@inproceedings{JurdzinskiL17,
  author    = {Marcin Jurdzi{\'{n}}ski and Ranko Lazi{\'{c}}},
  title     = {Succinct progress measures for solving parity games},
  booktitle = lics17,
  pages     = {1--9},
  year      = {2017},
  doi       = {10.1109/LICS.2017.8005092}
}

@article{FearnleyJKSSW19,
  author  = {John Fearnley and Sanjay Jain and Bart de Keijzer and Sven Schewe and Frank Stephan and Dominik Wojtczak},
  title   = {An ordered approach to solving parity games in quasi-polynomial time and quasi-linear space},
  journal = {Int. J. Softw. Tools Technol. Transf.},
  volume  = {21},
  number  = {3},
  pages   = {325--349},
  year    = {2019},
  doi     = {10.1007/s10009-019-00509-3}
}

@article{GimbertI17,
  author  = {Hugo Gimbert and Rasmus Ibsen{-}Jensen},
  title   = {A short proof of correctness of the quasi-polynomial time algorithm for parity games},
  journal = {CoRR},
  volume  = {abs/1702.01953},
  year    = {2017},
  url     = {http://arxiv.org/abs/1702.01953}
}

@inproceedings{Lehtinen2018modal,
  author       = {Karoliina Lehtinen},
  Aeditor       = {Anuj Dawar and
                  Erich Gr{\"{a}}del},
  title        = {A modal {\(\mu\)} perspective on solving parity games in quasi-polynomial
                  time},
  booktitle    = {Proceedings of the 33rd Annual {ACM/IEEE} Symposium on Logic in Computer
                  Science, {LICS} 2018, Oxford, UK, July 09-12, 2018},
  pages        = {639--648},
  publisher    = {{ACM}},
  year         = {2018},
  url          = {https://doi.org/10.1145/3209108.3209115},
  doi          = {10.1145/3209108.3209115},
  bibsource    = {dblp computer science bibliography, https://dblp.org}
}

@article{LehtinenPSW22,
  author  = {Karoliina Lehtinen and Pawe{\l} Parys and Sven Schewe and Dominik Wojtczak},
  title   = {A Recursive Approach to Solving Parity Games in Quasipolynomial Time},
  journal = {Log. Methods Comput. Sci.},
  volume  = {18},
  number  = {1},
  year    = {2022},
  doi     = {10.46298/lmcs-18(1:8)2022}
}

@inproceedings{FijalkowGO20,
  author    = {Nathana{\"{e}}l Fijalkow and Pawe{\l} Gawrychowski and Pierre Ohlmann},
  title     = {Value Iteration Using Universal Graphs and the Complexity of Mean Payoff Games},
  booktitle = {45th International Symposium on Mathematical Foundations of Computer Science, {MFCS} 2020},
  series    = {LIPIcs},
  volume    = {170},
  pages     = {34:1--34:15},
  publisher = {Schloss Dagstuhl - Leibniz-Zentrum f{\"{u}}r Informatik},
  year      = {2020},
  doi       = {10.4230/LIPIcs.MFCS.2020.34}
}

@inproceedings{chakrabarti2003resource,
  author    = {Arindam Chakrabarti and Luca de Alfaro and Thomas A. Henzinger and Mari{\"{e}}lle Stoelinga},
  title     = {Resource interfaces},
  booktitle = {International Workshop on Embedded Software},
  pages     = {117--133},
  year      = {2003},
  publisher = {Springer},
  doi       = {10.1007/978-3-540-45212-6_9}
}

@inproceedings{bouyer2008infinite,
  author    = {Patricia Bouyer and Uli Fahrenberg and Kim G. Larsen and Nicolas Markey and Ji{\v{r}}{\'{i}} Srba},
  title     = {Infinite runs in weighted timed automata with energy constraints},
  booktitle = {International Conference on Formal Modeling and Analysis of Timed Systems},
  pages     = {33--47},
  year      = {2008},
  publisher = {Springer},
  doi       = {10.1007/978-3-540-85778-5_4}
}

@article{brim2011faster,
  author  = {Lubo{\v{s}} Brim and Jakub Chaloupka and Laurent Doyen and Raffaella Gentilini and Jean-Fran{\c{c}}ois Raskin},
  title   = {Faster algorithms for mean-payoff games},
  journal = {Formal Methods in System Design},
  volume  = {38},
  number  = {2},
  pages   = {97--118},
  year    = {2011},
  doi     = {10.1007/s10703-010-0105-x}
}

@inproceedings{emerson1993model,
  author    = {E. Allen Emerson and Charanjit S. Jutla and A. Prasad Sistla},
  title     = {On model-checking for fragments of $\mu$-calculus},
  booktitle = {International Conference on Computer Aided Verification},
  pages     = {385--396},
  year      = {1993},
  publisher = {Springer},
  doi       = {10.1007/3-540-56922-7_32}
}

@article{mcnaughton1993infinite,
  author  = {Robert McNaughton},
  title   = {Infinite games played on finite graphs},
  journal = {Annals of Pure and Applied Logic},
  volume  = {65},
  number  = {2},
  pages   = {149--184},
  year    = {1993},
  doi     = {10.1016/0168-0072(93)90036-D}
}

@inproceedings{AkianCDG13,
  author    = {Marianne Akian and Jean Cochet{-}Terrasson and Sylvie Detournay and St{\'{e}}phane Gaubert},
  title     = {Solving multichain stochastic games with mean payoff by policy iteration},
  booktitle = {Proceedings of the 52nd {IEEE} Conference on Decision and Control, {CDC} 2013},
  pages     = {1834--1841},
  publisher = {{IEEE}},
  year      = {2013},
  doi       = {10.1109/CDC.2013.6760149}
}

@inproceedings{AugerMS21,
  author    = {David Auger and Xavier Badin de Montjoye and Yann Strozecki},
  title     = {A Generic Strategy Improvement Method for Simple Stochastic Games},
  booktitle = {46th International Symposium on Mathematical Foundations of Computer Science, {MFCS} 2021},
  series    = {LIPIcs},
  volume    = {202},
  pages     = {12:1--12:22},
  publisher = {Schloss Dagstuhl - Leibniz-Zentrum f{\"{u}}r Informatik},
  year      = {2021},
  doi       = {10.4230/LIPIcs.MFCS.2021.12}
}

@article{ColcombetFGO22,
  author  = {Thomas Colcombet and Nathana{\"{e}}l Fijalkow and Pawe{\l} Gawrychowski and Pierre Ohlmann},
  title   = {The Theory of Universal Graphs for Infinite Duration Games},
  journal = {Log. Methods Comput. Sci.},
  volume  = {18},
  number  = {3},
  year    = {2022},
  doi     = {10.46298/lmcs-18(3:29)2022}
}

@article{CominR17,
  author  = {Carlo Comin and Romeo Rizzi},
  title   = {Improved Pseudo-polynomial Bound for the Value Problem and Optimal Strategy Synthesis in Mean Payoff Games},
  journal = {Algorithmica},
  volume  = {77},
  number  = {4},
  pages   = {995--1021},
  year    = {2017},
  doi     = {10.1007/s00453-016-0123-1}
}

@inproceedings{CzerwinskiDFJLP19,
  author    = {Wojciech Czerwi{\'{n}}ski and Laure Daviaud and Nathana{\"{e}}l Fijalkow and Marcin Jurdzi{\'{n}}ski and Ranko Lazi{\'{c}} and Pawe{\l} Parys},
  title     = {Universal trees grow inside separating automata: Quasi-polynomial lower bounds for parity games},
  booktitle = {Proceedings of the Thirtieth Annual {ACM-SIAM} Symposium on Discrete Algorithms, {SODA} 2019},
  pages     = {2333--2349},
  publisher = {{SIAM}},
  year      = {2019},
  doi       = {10.1137/1.9781611975482.142}
}

@incollection{Everett1957,
  author    = {Hugh Everett},
  title     = {Recursive games},
  booktitle = {Contributions to the Theory of Games, Volume III},
  series    = {Annals of Mathematics Studies},
  volume    = {39},
  pages     = {67--78},
  year      = {1957},
  publisher = {Princeton University Press},
  address   = {Princeton, NJ},
}

@incollection{Gillette1957,
  author    = {Dean Gillette},
  title     = {Stochastic games with zero stop probability},
  booktitle = {Contributions to the Theory of Games, Volume III},
  series    = {Annals of Mathematics Studies},
  volume    = {39},
  pages     = {179--187},
  year      = {1957},
  publisher = {Princeton University Press},
  address   = {Princeton, NJ}
}

@article{liggett1969stochastic,
  author  = {Thomas M. Liggett and Steven A. Lippman},
  title   = {Stochastic Games with Perfect Information and Time Average Payoff},
  journal = {SIAM Review},
  volume  = {11},
  number  = {4},
  pages   = {604--607},
  year    = {1969},
  doi     = {10.1137/1011093}
}

@article{mertens1981stochastic,
  author  = {Jean-Fran{\c{c}}ois Mertens and Abraham Neyman},
  title   = {Stochastic games},
  journal = {International Journal of Game Theory},
  volume  = {10},
  number  = {2},
  pages   = {53--64},
  year    = {1981},
  doi     = {10.1007/BF01769259}
}

@article{bellman1957markovian,
  author  = {Richard E. Bellman},
  title   = {A {M}arkovian decision process},
  journal = {Journal of Mathematics and Mechanics},
  volume  = {6},
  number  = {5},
  pages   = {679--684},
  year    = {1957},
  doi     = {10.1512/iumj.1957.6.56038}
}

@book{sutton2018reinforcement,
  author    = {Richard S. Sutton and Andrew G. Barto},
  title     = {Reinforcement Learning: An Introduction},
  edition   = {2},
  year      = {2018},
  publisher = {MIT Press},
  address   = {Cambridge, MA}
}

@book{Teschl12,
  author    = {Gerald Teschl},
  title     = {Ordinary Differential Equations and Dynamical Systems},
  series    = {Graduate Studies in Mathematics},
  volume    = {140},
  publisher = {American Mathematical Society},
  address   = {Providence, RI},
  year      = {2012},
  doi       = {10.1090/gsm/140}
}

@book{Walter98,
  author    = {Wolfgang Walter},
  title     = {Ordinary Differential Equations},
  series    = {Graduate Texts in Mathematics},
  volume    = {182},
  publisher = {Springer},
  address   = {New York, NY},
  year      = {1998},
  doi       = {10.1007/978-1-4612-0601-9}
}

\end{document}